\documentclass[aps,twocolumn,showpacs,pra,citeautoscript,amsmath,amssymb,floatfix,superscriptaddress,10pt]{revtex4-1}
\pdfoutput=1
\usepackage{amsmath}
\usepackage{amssymb}
\usepackage{epsfig}
\usepackage{color}
\usepackage{amsthm}
\usepackage{hyperref}
\usepackage{appendix}
\usepackage{cancel}
\usepackage{stmaryrd}
\usepackage{soul}

\def\<{\langle}
\def\>{\rangle}

\newcommand{\be}{\begin{eqnarray} \begin{aligned}}
\newcommand{\ee}{\end{aligned} \end{eqnarray} }
\newcommand{\benn}{\begin{eqnarray*} \begin{aligned}}
\newcommand{\eenn}{\end{aligned} \end{eqnarray*} }

\newcommand{\ben}{\begin{eqnarray} \begin{aligned}}
\newcommand{\een}{\end{aligned} \end{eqnarray} }

\newcommand{\bc}{\begin{center}}
\newcommand{\ec}{\end{center}}

\newcommand{\tr}{\mathop{\mathsf{tr}}\nolimits}

\newcommand{\beq}{\begin{eqnarray} \begin{aligned}}
\newcommand{\eeq}{\end{aligned} \end{eqnarray} }
\newcommand{\bea}{\begin{array}}
\newcommand{\eea}{\end{array}}

\newcommand{\bee}{\begin{enumerate}}
\newcommand{\eee}{\end{enumerate}}
\newcommand{\bei}{\begin{itemize}}
\newcommand{\eei}{\end{itemize}}
\newtheorem{theorem}{Theorem}

\newtheorem{lemma}[theorem]{Lemma}

\usepackage{amsfonts}

\def\01{\{0,1\}}

\newcommand{\ket}[1]{|#1\rangle}
\newcommand{\bra}[1]{\langle#1|}

\newcommand{\ketbra}[2]{|#1\rangle\langle#2|}

\newcommand{\braket}[2]{\langle #1|#2\rangle}

\def\<{\langle}
\def\>{\rangle}

\newtheorem*{rep@theorem}{\rep@title}
\newcommand{\newreptheorem}[2]{%
\newenvironment{rep#1}[1]{%
 \def\rep@title{#2 \ref{##1} (restatement)}%
 \begin{rep@theorem}}%
 {\end{rep@theorem}}}
\makeatother

\newreptheorem{thm}{Theorem}
\newreptheorem{lem}{Lemma}

\begin{document}

\title{What is the probability of a thermodynamical transition?}

\author{\'{A}lvaro M. Alhambra}
\email{alvaro.alhambra.14@ucl.ac.uk}
\affiliation{Department of Physics and Astronomy, University College London, Gower Street, London WC1E 6BT, United Kingdom}
\author{Jonathan Oppenheim}
\email{j.oppenheim@ucl.ac.uk}
\affiliation{Department of Physics and Astronomy, University College London, Gower Street, London WC1E 6BT, United Kingdom}
\affiliation{Department of Computer Science and Centre for Quantum Technologies, National University of Singapore, Singapore 119615}
\author{Christopher Perry}
\email{christopher.perry.12@ucl.ac.uk}
\affiliation{Department of Physics and Astronomy, University College London, Gower Street, London WC1E 6BT, United Kingdom}

\begin{abstract}
If the second law of thermodynamics forbids a transition from one state to another, then it is still possible to make the transition happen by using a sufficient amount of work. But if we do not have access to this amount
of work, can the transition happen probabilistically?
In the thermodynamic limit, this probability tends to zero, but here we find that for finite-sized systems, it can be finite. We compute the maximum probability
of a transition or a thermodynamical fluctuation from any initial state to any final state,  and show that this maximum can be achieved
for any final state which is block-diagonal in the energy eigenbasis. 
We also find upper and lower bounds on this transition probability, in terms of the work of transition. As a bi-product,
we introduce a finite set of thermodynamical monotones related to the thermo-majorization criteria which governs state transitions, and compute the work of transition in terms of them. The trade-off between the probability of a transition, and any partial work added to aid in that transition is also considered.
Our results 
have applications in entanglement theory, and we find the amount of entanglement required (or gained) when transforming one pure entangled state into any other.
\end{abstract}
\maketitle

\section{Introduction}

Given a quantum system in a state $\rho$ with some Hamiltonian, $H_1$, when can it be deterministically transformed into another state $\sigma$ associated with a potentially different Hamiltonian, $H_2$? 
If we can put the system into contact with a heat bath at temperature $T$, then in the thermodynamical limit, and if interactions are short-ranged or screened, a transition will occur as long as the free energy of the initial configuration is larger than the free energy of the final
configuration. The free energy of the state $\rho$ defined as:
\begin{align}
F(\rho,H_1)=\tr \left[H_1\rho\right] - TS\left(\rho\right),
\end{align} 
were $S(\rho)$ is the entropy; $S(\rho)=-\tr\rho\log\rho$. This is a formulation of the second law of thermodynamics, if we factor in energy conservation (the first law). If we wish to make a forbidden transition occur, then we need
to inject an amount of work which is greater than the free energy difference between initial and final states.

However, what if we are 
interested in small, finite-sized systems? Or in systems with long-range interactions? The thermodynamics of systems in the micro-regime, where we do not take the thermodynamical limit, has gained increased importance as we cool 
and manipulate smaller and smaller systems at the nano scale and beyond~\cite{Scovil1959masers,scully2002afterburner,rousselet1994directional,Faucheux1995ratchet,baugh2005experimental}. Theoretical work has continued a pace, with increased interest in the field in recent years~\cite{ruch1975diagram,ruch1976principle,ruch1978mixing,Geusic1967quatum,Alicki79,howard1997molecular,geva1992classical,Hanggi2009brownian,allahverdyan2000extraction,uniqueinfo,feldmann2006lubrication,linden2010small,dahlsten2011inadequacy,del2011thermodynamic, HO-limitations, aaberg-singleshot,  faist2015minimal, skrzypczyk2014work,egloff2015measure, brandao2013second, cwiklinski2015limitations, lostaglio2015description, lostaglio2014quantum, halpern2014beyond, wilming2014weak, lostaglio2015stochastic, narasimhachar2015low}.
If we do not take the thermodynamical limit, then
provided $\sigma$ is block-diagonal in the energy eigenbasis, there is not just one criteria (the decreasing of the free energy), but a family of criteria which
determine whether a state transition is possible. A set of such criteria
which have been proven to be necessary and sufficient condition for quantum thermodynamical state transformations~\cite{HO-limitations} (c.f. \cite{ruch1978mixing}), are the 
so-called thermo-majorization criteria~\cite{ruch1976principle,HO-limitations}.
Thermo-majorization is a set of conditions that are
more stringent than the ordinary second laws and had been conjectured to provide a limitation on the possibility of thermodynamical transformations since 1975~\cite{ruch1976principle}. It is related
\cite{ruch1978mixing,renes2014work} to a condition known as Gibbs-stochasity\cite{Streater_dynamics,janzing2000thermodynamic} a condition which can be extended to include fluctuations of work \cite{alhambra2016second}. 

Once again though, if the diagonal state $\sigma$, is not thermo-majorized by $\rho$, then a transition is still possible, provided sufficient work is used.
One can compute the work required (or gained) from this transition using thermo-majorization diagrams \cite{HO-limitations}, via a linear program\cite{renes2014work}, or the relative-mixedness \cite{egloff2015measure}.
Suppose however, we want to make a transition from $\rho$ to $\sigma$, and it requires work which we cannot, or do not wish to, expend. 
Can we still nonetheless make the transition with some probability $p$ rather than with certainty? And if so, what is the highest probability, $p^*$, that can be achieved? In particular, 
given $\rho$ and $\sigma$, we are interested in maximizing $p$ in the following process:
\begin{equation}
\rho{\longrightarrow}\rho'= p\sigma+\left(1-p\right)X, \label{Goal}
\end{equation}
with $X$ being some arbitrary state. 

Such a transformation can be regarded as a fluctuation of a system's state, in the sense that the transformation is only probabilistic. Within the study of thermodynamics for small systems, great progress has already been made in analyzing how the work distribution associated with a given transformation of process can fluctuate \cite{jarzynski1997nonequilibrium,crooks1999entropy,tasaki2000jarzynski,talkner2009fluctuation} (see \cite{seifert2012stochastic,esposito2009nonequilibrium,campisi2011colloquium} for reviews on both the classical and quantum cases). Fluctuation relations such as the Jarzynski equality \cite{jarzynski1997nonequilibrium} and Crooks' theorem \cite{crooks1999entropy}, developed under the paradigm of stochastic thermodynamics, have been used to calculate the work fluctuations of non-equilibrium processes. Investigating fluctuation in a system's state provides a natural, complementary strand of research which we are able to formulate and analyze in this paper by applying techniques from quantum information theory developed in \cite{HO-limitations}. In related work \cite{alhambra2016second}, we shall address the problem of fluctuating work within this information theoretic framework. This shall serve to unify the two approaches to thermodynamics for small systems and extend and provide insight into previous work based on the stochastic thermodynamics perspective.

Here, we will upper bound the maximum probability of a fluctuation between any given $\rho$ and $\sigma$. When $\sigma$ is block-diagonal in the energy eigenbasis, we will show that this bound can be achieved and furthermore, that there exists a two outcome measurement that can be performed on 
$\rho'$ such that we obtain $\sigma$ with the maximum probability $p^*$. Of course, measurements do not come for free in thermodynamics - it costs work to erase the record of the measurement outcome \cite{bennett1982thermodynamics}. That this measurement can be performed
is noted for completeness -- however, we take Eq. \eqref{Goal} as our primary goal, defining what we mean by a thermodynamical transition. We will discuss measurements in Section \ref{sec:measurement} as they only provide a small correction
of $kT\log2$ to the work cost of a probabilistic transformation.

Our main result will be Theorem \ref{pthermal}, which upper bounds the probability $p^*$ in terms of a minimization over a finite set of ratios between thermodynamical monotones, which are quantities that can only decrease under the set of allowed operations. When the final state is block-diagonal, this bound is achievable, but this may not be the case if the final state has coherences in energy. These monotones, which we will show are given by Eq.
\eqref{eq:TOmonotones}, can be thought of as analogous to free energies.
This is proven in  
Theorem \ref{Finite set theorem} and is equivalent to the thermo-majorization criteria  of \cite{HO-limitations,ruch1976principle}. The set of ratios that we use to bound $p^*$ thus gives an alternative way of verifying if the thermo-majorization criteria are satisfied. Rather than considering the thermo-majorization curves \cite{HO-limitations} or considering a continuous set of monotones~\cite{egloff2015measure} we provide a finite set of conditions to check. Indeed this set provides a strengthening of results from the theory of relative majorization \cite[14.B.4(c)]{marshall2010inequalities} by reducing the number of constraints that need to be considered. 

Before proving Theorem \ref{pthermal}, we will consider in Section \ref{sec:NO} the simpler case where the Hamiltonian of the system is trivial, i.e. $H\propto\mathbb{I}$.  Solving the problem in this regime, referred to as Noisy Operations \cite{uniqueinfo,gour2015resource}, will provide us with insight into the solution for non-trivial Hamiltonians. In this simplified situation, $p^*$ is given by Theorem \ref{NO Theorem}.
The result is similar in form to \cite{vidal1999entanglement} which considers the analogous problem of probabilistic pure state entanglement manipulation using Local Operations and Classical Communication (LOCC). 
However, care must be taken -- the class of operations allowed under LOCC is very different to what is allowed in thermodynamics. For example, under LOCC one can bring in pure states for free (which can be a source of work in thermodynamics) and one is allowed to make measurements for free (which costs work). Perhaps more importantly, many of the LOCC monotones are concave, which is not the case in Noisy Operations, thus we will require some different techniques. It should also be noted that in entanglement manipulation, the maximum probability achievable will be zero if the target state has a larger Schmidt rank than the starting state. Under Noisy Operations, we will see that $p^*$ is always non-zero (though it can be arbitrarily small).

In Section \ref{sec:TO} we consider the general case of arbitrary initial and final Hamiltonians and states. We will prove our results using the paradigm of Thermal Operations (TO) \cite{Streater_dynamics,janzing2000thermodynamic,HO-limitations}.
There are a number of different paradigms one can use to study thermodynamics (e.g. allowing interaction Hamiltonians or changing energy levels), however, these other paradigms are equivalent to 
Thermal Operations \cite{brandao2013resource,HO-limitations}, and thus Thermal Operations are the appropriate paradigm for studying fundamental limitations.
We introduce Thermal Operations at the beginning of Section \ref{sec:TO}. In the case of a trivial Hamiltonian, Thermal Operations reduce to Noisy Operations, the regime considered in Section
\ref{sec:NO}.

Our expression for the cost of a transition between any two states using only a finite number of monotones is given in Lemma \ref{NO Work Exp} for Noisy Operations and Lemma \ref{TO Work Exp} for Thermal Operations. The Noisy Operations result can be adapted to give an expression for the amount of entanglement required (or gained) when transforming any pure bipartite state into another under LOCC. This is given in Appendix \ref{ap:LOCC} and generalizes existing expressions for the distillable entanglement \cite{buscemi2010distilling,buscemi2010general} and cost of entanglement formation \cite{buscemi2011entanglement}.
We also show how $p^*$ can be upper and lower bounded using the work of transitions from $\rho$ to $\sigma$ and $\sigma$ to $\rho$. This is done in Lemma \ref{bounds} for the case of a trivial Hamiltonian, and in Lemma \ref{thermobounds} for the general case.

Finally, we conclude in Section \ref{sec:conclusion} with a discussion on other goals, related to Eq. \eqref{Goal}, which one could attempt when making a probabilistic transition. 
One such goal, the optimization of the \emph{heralded probability}, is discussed in detail in Appendix \ref{ap:herald} where we obtain bounds on it, even in the presence of coherence or catalysts. 
The heralding probability can be thought of as a generalization of the case where one achieves Eq. \eqref{Goal} with a measurement i.e.
\begin{equation}
\rho\otimes\ketbra{0}{0}\stackrel{TO}{\longrightarrow}\hat{\rho}=p\sigma\otimes\ketbra{0}{0}+\left(1-p\right)X\otimes\ketbra{1}{1}.
\nonumber
\end{equation}
and the transition is conclusive.
This allows us to analyze state fluctuations in the presence of measurements, coherence and catalysis. 
We also pose some open questions. One of these regards how $p^*$ varies if we supply additional work to drive the transition from $\rho$ to $\sigma$ or demand that additional work be extracted. The solution for qubit systems with trivial Hamiltonian is given in Appendix \ref{ap:tradeoff}.

\section{Probability of transition under Noisy Operations}
\label{sec:NO}
Before investigating Eq. (\ref{Goal}) in the context of Thermal Operations, we will first consider a simpler, special case - Noisy Operations. In this particular instance of thermodynamics, the Hamiltonian of the system under consideration is trivial.  Noisy Operations were first defined in \cite{uniqueinfo} where the problem of whether a transition between two given states under a particular set of operations was considered. Within Noisy Operations, the following actions are allowed: \emph{i)} a system of any dimension in the maximally mixed state can be added, \emph{ii)} any subsystem can be discarded through tracing out and \emph{iii)} any unitary can be applied to the global system. Throughout this paper, we shall use $\eta_i$ to denote the eigenvalues of $\rho$ and $\zeta_i$ to denote those of $\sigma$. For a comprehensive review of Noisy Operations, see \cite{gour2015resource}.

Given two states, $\rho$ and $\sigma$, it was shown in \cite{uniqueinfo} that transition from $\rho$ to $\sigma$ is possible under Noisy Operations if and only if $\rho$ majorizes $\sigma$ (written $\rho\succ\sigma$). That is, if we list the eigenvalues of $\rho$ and $\sigma$ \footnote{Note that we can always assume that $\rho$ and $\sigma$ have the same number of eigenvalues. If they do not, by applying operation $\emph{i)}$ of Noisy Operations appropriately, we can ensure the systems under consideration have the same dimension.} in decreasing order and denote these ordered lists by $\vec{\eta}=\{\eta_1,\hdots,\eta_n\}$ and $\vec{\zeta}=\{\zeta_1,\hdots,\zeta_n\}$ respectively, the transition is possible if and only if:
\begin{equation}\label{majorization}
V_l(\rho) \geq V_l(\sigma), \quad \forall l \in \{1,\hdots,n\},
\end{equation}
where:
\begin{equation} \label{monotones}
V_l\left(\rho\right)=\sum_{i=1}^{l} \eta_i.
\end{equation}
Lorenz curves are a useful tool for visualizing these criteria (Figure \ref{Lorenz}). For a given state $\rho$, its Lorenz curve is formed by plotting the points:
\begin{equation}
\left\{\left(\frac{k}{n},\sum_{i=1}^{k}\eta_i\right)\right\}_{k=1}^{n},
\end{equation}
and connecting them piecewise linearly (together with the point $(0,0)$) to form a concave curve. If $\rho$ majorizes $\sigma$, the Lorenz curve for $\rho$ is never below that of $\sigma$.

\begin{figure}
\centering
\includegraphics[width=1\columnwidth]{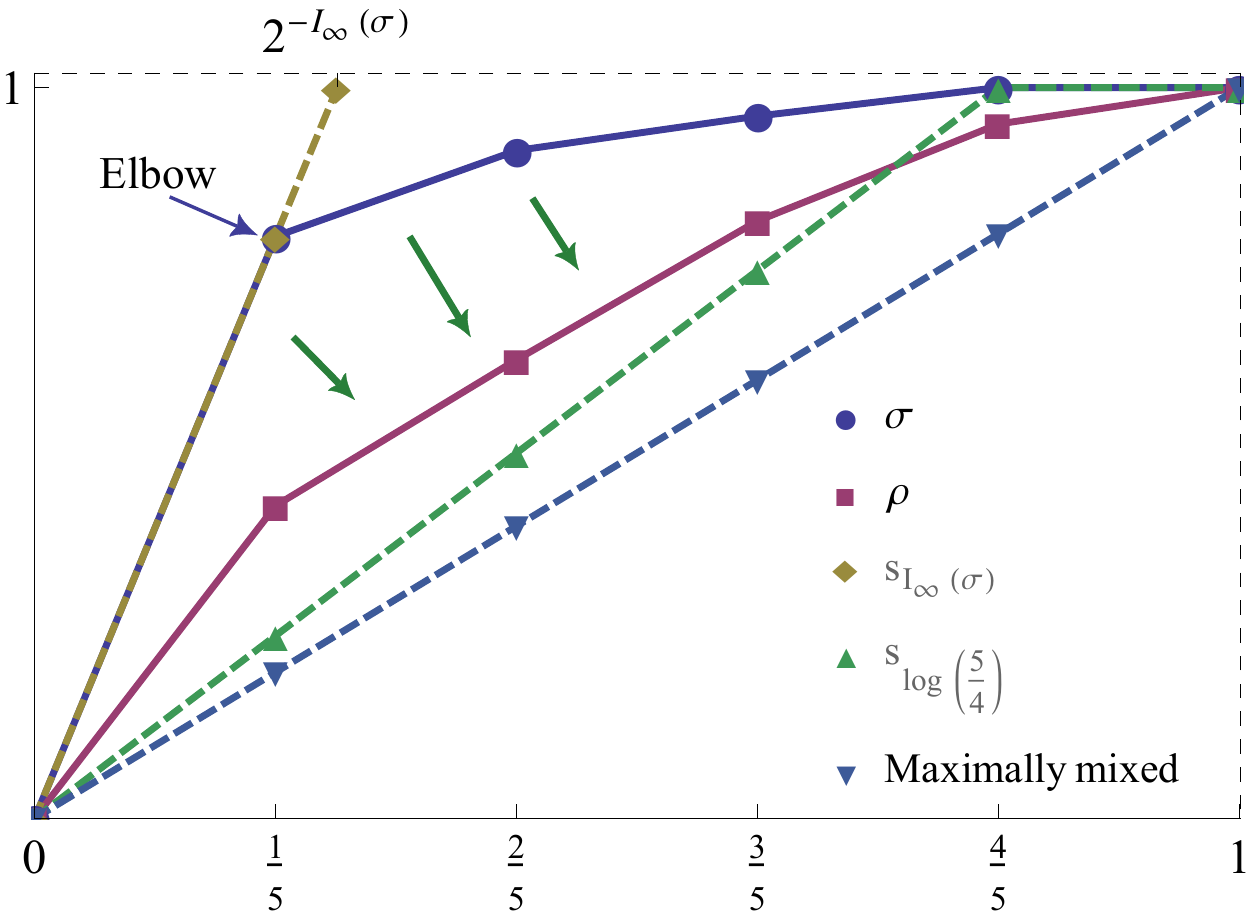}
\caption{ Lorenz Curves. \emph{a)} The Lorenz curve for $\rho$ is defined by plotting the points: $\left\{\left(\frac{k}{n},\sum_{i=1}^{k}\eta_i\right)\right\}_{k=1}^{n}$. \emph{b)} The transition from $\sigma$ to $\rho$ is possible under NO as the curve for $\sigma$ is never below that of $\rho$. \emph{c)} The Lorenz curve for a maximally mixed state is given by the dashed line from $(0,0)$ to $(1,1)$. All other states majorize it. \emph{d)} $s_{\log\frac{5}{4}}$ is an example of a sharp state. \emph{e)} $s_{I_{\infty}\left(\sigma\right)}$ is the least sharp state that majorizes $\sigma$.} \label{Lorenz}
\end{figure}

The functions defined in Eq. (\ref{monotones}), and their analogue in Thermal Operations, will be crucial for the rest of the paper. They are monotones of the theory, only decreasing under Noisy Operations. Excellent reviews regarding the theory of majorization and Lorenz curves can be found in \cite{marshall2010inequalities, gour2015resource}.

\subsection{Non-deterministic transitions}

We will now consider transitions when the conditions given in Eq. \eqref{majorization} are not necessarily fulfilled. Here, rather than transforming $\rho$ to $\sigma$ with certainty, we shall do so with some probability as formulated in Eq. \eqref{Goal}. In particular, we are interested in the maximum probability, $p^*$, that can be achieved. A similar problem is considered in \cite{vidal1999entanglement} for entanglement manipulation and adapting its techniques the following theorem can be shown:
\begin{theorem} \label{NO Theorem}
Suppose we wish to transform the state $\rho$ to the state $\sigma$ under Noisy Operations. The maximum value of $p$ that can be achieved in the transition:
\begin{equation} \label{NO Goal}
\rho\stackrel{\textit{NO}}{\longrightarrow}\rho'= p\sigma+\left(1-p\right)X,
\end{equation}
is given by:
\begin{equation} \label{NO p*}
p^{*}=\min_{l\in\{1,\hdots,n\}} \frac{V_l\left(\rho\right)}{V_l\left(\sigma\right)}.
\end{equation}
\end{theorem}
\begin{proof}

The proof is split into two parts: first we apply Weyl's inequality and the definition of majorization to derive a contradiction if it were possible to achieve a value of p large than $p^*$. Next, we adapt the techniques of \cite{vidal1999entanglement} to provide a protocol achieving $p=p^*$.

To achieve our first goal we begin by showing that given Eq. \eqref{NO Goal}:
\begin{equation}
V_l\left(\rho\right)\geq p V_l\left(\sigma\right), \quad\forall l. \label{prob mon}
\end{equation}

To prove this, we make use of Weyl's inequality \cite{weyl1912asymptotische, horn2012matrix}. Given $n \times n$ Hermitian matrices, $A$, $B$ and $C$ such that $A=B+C$, let $\left\{a_i\right\}_{i=1}^{n}$, $\left\{b_i\right\}_{i=1}^{n}$ and $\left\{c_i\right\}_{i=1}^{n}$ be their respective eigenvalues arranged in descending order. Weyl's inequality then states that:
\begin{equation}
b_i+c_n \leq a_i \leq b_i+c_1,
\end{equation}
for all $i$. Applying this to $\rho'$, $\sigma$ and $X$, we obtain:
\begin{equation}
\eta'_i\geq p\zeta_i + \left(1-p\right) x_n, \quad \forall i,
\end{equation}
where $x_n$ is the smallest eigenvalue of $X$. As $X$ is a positive semidefinite matrix, $x_n\geq0$ and:
\begin{equation} \label{evalue bound}
\eta'_i\geq p\zeta_i, \quad \forall i.
\end{equation}
Hence:
\begin{equation}
V_l\left(\rho\right)\geq V_l\left(\rho'\right)
=\sum_{i=1}^l \eta_i'
\geq p\sum_{i=1}^{l} \zeta_i
=p V_l\left(\sigma\right),
\end{equation}
where the first inequality uses Eq. \eqref{majorization} and the second follows from Eq. \eqref{evalue bound}.

Now suppose it was possible to achieve a value of $p$ greater than $p^{*}$ in Eq. \eqref{NO p*}. Then there would exist an $l$ such that $V_l\left(\rho\right)<p V_l\left(\sigma\right)$, contradicting Eq. \eqref{prob mon}.

To show that $p^*$ is obtainable, we define the following quantities. First, define $l_1$ by:
\begin{equation}
l_1=\max\left\{l:\frac{V_l\left(\rho\right)}{V_l\left(\sigma\right)}=p^{*}\equiv r^{\left(1\right)}\right\}.
\end{equation}
Then we proceed iteratively and, provided $l_{i-1}< n$, define:
\begin{equation}
r^{\left(i\right)}=\min_{l>l_{i-1}}\frac{V_l\left(\rho\right)-V_{l_{i-1}}\left(\rho\right)}{V_l\left(\sigma\right)-V_{l_{i-1}}\left(\sigma\right)},
\end{equation}
so we have:
\begin{equation}
r^{\left(i\right)}\sum_{j=l_{i-1}+1}^{l} \zeta_j \leq \sum_{j=l_{i-1}+1}^{l} \eta_j, \quad \forall l>l_{i-1}. \label{NO Block Cond}
\end{equation}
Define $l_{i}$ by:
\begin{equation}
l_{i}=\max\left\{l:l>l_{i-1}, \frac{V_l\left(\rho\right)-V_{l_{i-1}}\left(\rho\right)}{V_l\left(\sigma\right)-V_{l_{i-1}}\left(\sigma\right)}=r^{\left(i\right)} \right\}.
\end{equation}

Note that we have $r^{(i)}>r^{(i-1)}$. To see this, first observe that for $a,b,c,d>0$:
\begin{equation} \label{eq:abcd}
\frac{a}{b}<\frac{a+c}{b+d}\Leftrightarrow \frac{a}{b}<\frac{c}{d}.
\end{equation}
Setting:
\begin{align*}
a&=V_{l_{i-1}}\left(\rho\right)-V_{l_{i-2}}\left(\rho\right),				\\
b&=V_{l_{i-1}}\left(\sigma\right)-V_{l_{i-2}}\left(\sigma\right),				\\
c&=V_{l_{i}}\left(\rho\right)-V_{l_{i-1}}\left(\rho\right),				\\
d&=V_{l_{i}}\left(\sigma\right)-V_{l_{i-1}}\left(\sigma\right)				,
\end{align*}
so $\frac{a}{b}=r^{\left(i-1\right)}$ and $\frac{c}{d}=r^{\left(i\right)}$, then:
\begin{equation*}
\frac{a+c}{b+d}=\frac{V_{l_{i}}\left(\rho\right)-V_{l_{i-2}}\left(\rho\right)}{V_{l_{i}}\left(\sigma\right)-V_{l_{i-2}}\left(\sigma\right)}>r^{\left(i-1\right)}=\frac{a}{b},
\end{equation*}
where the inequality follows from the definition of $r^{\left(i-1\right)}$. Using Eq.~\eqref{eq:abcd}, the claim that $r^{(i)}>r^{(i-1)}$ now follows. Overall, this protocol generates a set of $l_i$ such that $0=l_0<l_1<\hdots<l_k=n$ and a set of $r_i$ such that $p^{*}=r^{\left(1\right)}<\hdots<r^{\left(k\right)}$.


Now we split $\rho$ and $\sigma$ into blocks and define:
\begin{align}
\rho_i&=\textrm{diag}\left(\eta_{l_{i-1}+1},\hdots,\eta_{l_i}\right),\\
\sigma_i&=\textrm{diag}\left(\zeta_{l_{i-1}+1},\hdots,\zeta_{l_i}\right).
\end{align}
Then from Eq. (\ref{NO Block Cond}) (and the fact that equality occurs when $l=l_i$), $\rho_i$ majorizes $r^{\left(i\right)}\sigma_i$ and we can perform:
\begin{equation}\label{blocks}
\rho_i\stackrel{NO}{\longrightarrow}r^{\left(i\right)}\sigma_i=p^{*}\sigma_i+\left(r^{\left(i\right)}-p^{*}\right) \sigma_i, \quad \forall i.
\end{equation}
With a bit of massaging and recombining the blocks, this is the same form as Eq. (\ref{NO Goal}) with $p=p^*$ and the blocks of $X$ being defined by:
\begin{equation}
X_i=\frac{r^{\left(i\right)}-p^{*}}{1-p^{*}}\sigma_i.
\end{equation}
\end{proof}
Note that as the endpoints of the Lorenz curves coincide at $(1,1)$ and $\eta_1>0$, we are guaranteed that $0< p^* \leq 1$.

If we want to obtain $\sigma$ from $\rho$ with probability $p^*$ rather than have it as part of a probabilistic mixture  as per Eq. \eqref{NO Goal}, we can do so by performing a two outcome measurement, with measurement operators $\{\sqrt{M},\sqrt{\mathbb{I}-M}\}$, where the blocks of $M$ are given by:
\begin{equation} \label{POVM}
M_i=\textrm{diag}\left(\frac{p^{*}}{r^{\left(i\right)}},\hdots,\frac{p^{*}}{r^{\left(i\right)}}\right).
\end{equation}
To see that $M$ is a valid measurement, we note that in general $0 < \frac{p^{*}}{r^{\left(i\right)}}\le 1$. Hence both $\{\sqrt{M},\sqrt{\mathbb{I}-M}\}$ are well defined, and their squares trivially add up to the identity.

After applying this measurement to $\rho'$ and reading the result, we will have either:
\begin{equation}
\sqrt{M^{\phantom{\dagger}}} \rho' \sqrt{M^\dagger} = p^* \sigma,
\end{equation}
or 
\begin{equation}
\sqrt{\left(\mathbb{I}-M\right)^{\phantom{\dagger}}} \rho' \sqrt{\left(\mathbb{I}-M\right)^\dagger}=\left(1-p^*\right)X.
\end{equation}
However, performing this measurement is outside of the class of Noisy Operations and hence costs work. As such, if a general two outcome measurement is allowed without taking its cost into account, it can be possible to transform $\rho$ into $\sigma$ with probability greater than $p^*$. For example, if $\rho$ and $\sigma$ are qubits, we can convert $\rho$ into $\sigma$ with certainty using this extra resource. Firstly we add an additional qubit in the maximally mixed state and measure it in the computational basis. This results in a pure state, either $\ket{0}$ or $\ket{1}$. As these majorize all other qubit states we can use it to obtain any $\sigma$ with certainty.

\subsection{Nonuniformity of transition under Noisy Operations}\label{work}

If it is not possible to deterministically convert $\rho$ into $\sigma$ using Noisy Operations, to perform the transformation with certainty will cost some resource, in the form of {\it nonuniformity}. For instance, if we add some pure
states of sufficiently high dimension, a previously impossible transition will become possible. Adding these additional pure states can be thought of as the analogue to adding work. Similarly, if $\rho$ can be converted into $\sigma$ using Noisy Operations, it may be possible to extract some nonuniformity (e.g. by transforming some maximally mixed states into pure states). This is the analogue of extracting work. More generally, we shall extract or expend the equivalent of work using {\it sharp states}. These sharp states,  as discussed in the next subsection, will serve as a natural unit for the nonuniformity resource. We will compute the nonuniformity of transition in terms of a finite set of ratios of monotones. This is done in a similar manner to \cite{gour2015resource}, although we show that the minimization can be done over fewer points.

\subsubsection{Sharp States}

Quantifying the optimal amount of work of transition for the more general Thermal Operations was considered in \cite{HO-limitations, egloff2015measure}. We shall denote the Noisy Operations equivalent of work, the nonuniformity of transition, by $I_{\rho\rightarrow\sigma}$ . If nonuniformity must be added, the quantity is negative, while if we can extract nonuniformity, it will be positive. For $|I_{\rho\rightarrow\sigma}|=\log{\frac{d}{j}}$, we define an associated \emph{sharp state} \cite{gour2015resource} by:
\begin{equation}\label{sharp}
s_{|I_{\rho\rightarrow\sigma}|}=\textrm{diag}\biggl(\underbrace{\frac{1}{j},\hdots,\frac{1}{j}}_{j},\underbrace{0,\hdots,0\vphantom{\frac{1}{j}}}_{d-j}\biggr).
\end{equation}
Appending a sharp state $I_{\log{\frac{d}{j}}}$ to the system is equivalent to introducing $\log{\frac{d}{j}}$ units of nonuniformity.
See Figure \ref{Lorenz} for an example of a sharp state's Lorenz curve. The state $s_{|I_{\rho\rightarrow\sigma}|}$ is such that:
\begin{align}\label{workdef}
\begin{array}{rclc}
\rho\otimes s_{|I_{\rho\rightarrow\sigma}|}&\stackrel{NO}{\longrightarrow}&\sigma,&\quad\textrm{if }I_{\rho\rightarrow\sigma} \le 0,\\
\rho&\stackrel{NO}{\longrightarrow}&\sigma\otimes s_{|I_{\rho\rightarrow\sigma}|},&\quad\textrm{if }I_{\rho\rightarrow\sigma}>0.
\end{array}
\end{align}
In terms of Lorenz curves, tensoring a state $\rho$ with a sharp state $s_I$ has the effect of compressing the Lorenz curve of $\rho$ by a factor of $2^{-I}$ with respect to the $x$-axis \cite{HO-limitations}.

\subsubsection{Monotones for Noisy Operations and the nonuniformity of transition}

The function $V_l\left(\rho\right)$ is equal to the height of the Lorenz curve of $\rho$ at $x=\frac{l}{n}$. An alternative set of monotones, $L_{y}\left(\rho\right)$ where $0\leq y\leq 1$, can be defined as the shortest horizontal distance between the Lorenz curve of $\rho$ and the $y$-axis at $y$. Note that these functions \emph{never decrease} under Noisy Operations. In particular:
\begin{align}
\begin{split}
L_{y_k}\left(\rho\right)&=\frac{k}{n}, \textrm{ for } y_k=\sum_{i=1}^{k}\eta_i, \quad 1\leq k< \textrm{rank}\left(\rho\right),\\
L_{1}\left(\rho\right)&=\frac{\textrm{rank}\left(\rho\right)}{n}.
\end{split}
\end{align}
If we define the set $\mathcal{D}\left(\sigma\right)$ by:
\begin{equation}
\mathcal{D}\left(\sigma\right)=\left\{\sum_{i=1}^{k}\zeta_i\right\}_{k=1}^{\textrm{rank}\left(\sigma\right)},
\end{equation}
then a transition from $\rho$ to $\sigma$ is achievable with certainty under Noisy Operations if and only if:
\begin{equation} \label{NO Work Cond}
L_{y}\left(\rho\right)\leq L_{y}\left(\sigma\right), \quad \forall y\in \mathcal{D}\left(\sigma\right). 
\end{equation}
That it is sufficient to consider only $y\in\mathcal{D}\left(\sigma\right)$ will be justified below.

The horizontal monotones, $L_{y}$, also allow us to quantify the optimal work of transition that is required or extracted in going from $\rho$ to $\sigma$:
\begin{lemma} \label{NO Work Exp}
Given two states $\rho$ and $\sigma$, under Noisy Operations:
\begin{equation}
2^{-I_{\rho\rightarrow\sigma}}=\max_{y\in\mathcal{D}\left(\sigma\right)}\frac{L_y\left(\rho\right)}{L_y\left(\sigma\right)}.
\end{equation}
\end{lemma}
\begin{proof}
To prove this, we make use of the geometrical structure of Lorenz curves and the properties of $I_{\rho \rightarrow \sigma}$.
Note that we have:
\begin{equation} \label{W1}
2^{-I_{\rho\rightarrow\sigma}}=\max_{y\in\left[0,1\right]}\frac{L_y\left(\rho\right)}{L_y\left(\sigma\right)},
\end{equation}
as this follows from the fact that to obtain the optimal value of $I_{\rho\rightarrow\sigma}$, we wish to rescale the Lorenz curve of $\rho$ with respect to the $x$-axis in such a way that it just majorizes that of $\sigma$ - the curves should touch but not cross. The amount that we need to rescale by is given by Eq. \eqref{W1}.

We now show that it is sufficient to maximize over $y\in\mathcal{D}\left(\sigma\right)$. Let $s_0=0$ and $s_k=\sum_{i=1}^{k}\zeta_i$ for $1\leq k\leq\textrm{rank}\left(\sigma\right)$. Then, for $1\leq j\leq\textrm{rank}\left(\sigma\right)$, as the Lorenz curve of $\sigma$ is a straight line on the interval $\left[s_{j-1},s_{j}\right]$ and the Lorenz curve of $\rho$ is concave:
\begin{align} \label{interval}
\max_{y\in\left[s_{j-1},s_{j}\right]} \frac{L_y\left(\rho\right)}{L_y\left(\sigma\right)} \leq \max_{r\in\left[0,1\right]} \frac{rL_{s_{j-1}}\left(\rho\right)+\left(1-r\right)L_{s_{j}}\left(\rho\right)}{r\frac{j-1}{n}+\left(1-r\right)\frac{j}{n}}.
\end{align}
It is straightforward to check that the maximum value occurs at either $r=0$ or $r=1$. We can thus replace the inequality in Eq. \eqref{interval} with an equality and it follows that it suffices to maximize over $y\in\mathcal{D}\left(\sigma\right)$.
\end{proof}
As $\rho\stackrel{\textit{NO}}{\longrightarrow}\sigma$ is possible if and only if $I_{\rho\rightarrow\sigma}\geq0$, the finite set in Eq. \eqref{NO Work Cond} is justified.

Note that in \cite{gour2015resource} it was shown that it is possible to calculate $I_{\rho\rightarrow\sigma}$ by performing an optimization over the ratios calculated at the `elbows' (see Figure \ref{Lorenz} for a definition)  of both $\rho$ and $\sigma$. In Lemma \ref{NO Work Exp} we have shown that it suffices to consider just the `elbows' of $\sigma$.

\subsubsection{Bounds on the transition probability}
\label{ss:bounds}

The quantities $I_{\rho\rightarrow\sigma}$ and $I_{\sigma\rightarrow\rho}$ can be used to bound $p^{*}$ as follows:

\begin{lemma} \label{bounds}
Given two states $\rho$ and $\sigma$, under Noisy Operations:
\begin{equation}\label{bounds1}
2^{I_{\rho\rightarrow\sigma}}\leq p^{*} \leq 2^{-I_{\sigma\rightarrow\rho}},
\end{equation}
where as $p^{*}\leq 1$, we assume $I_{\rho\rightarrow\sigma}\leq0$. If $I_{\rho\rightarrow\sigma}\geq0$, $p^*=1$ and the transformation from $\rho$ to $\sigma$ can be done deterministically, potentially extracting a finite amount of nonuniformity.
\end{lemma}
\begin{proof}
We start proving with the lower bound, giving a protocol which achieves $p=2^{W_{\rho\rightarrow\sigma}}$. The upper bound is derived by considering properties of the purity of the least sharp state that majorizes $\rho$.

Assuming $|W_{\rho\rightarrow\sigma}|=\log{\frac{d}{j}}$ for simplicity, and defining $\mathbb{I}_{d}$ to be the maximally mixed state of a $d$ level system:
\begin{align}\label{eq:Ydef}
\rho\stackrel{NO}{\longrightarrow}&\rho\otimes \mathbb{I}_{d}, \\ \nonumber
=&\frac{j}{d}\rho\otimes s_{\log{\frac{d}{j}}}+\frac{d-j}{d}\rho\otimes s_{\log{\frac{d}{d-j}}},\\ \nonumber
\stackrel{NO}{\longrightarrow}&\frac{j}{d}\sigma\otimes \mathbb{I}_{d}+ \frac{d-j}{d} Y,\\ \nonumber
\stackrel{NO}{\longrightarrow}&\frac{j}{d}\sigma + \frac{d-j}{d} \textrm{Tr}_B Y, \nonumber
\end{align}
where $Y$ is the state obtained by applying the second Noisy Operation to $\rho\otimes s_{\log{\frac{d}{d-j}}}$. Using this protocol, we obtain something of the form Eq. (\ref{NO Goal}) with $p=2^{I_{\rho\rightarrow\sigma}}$ and $X=\textrm{Tr}_B Y$. As $p^{*}$ is the maximum value of $p$ obtainable in Eq. \eqref{NO Goal}, we derive the lower bound.

We now consider the upper bound and to obtain a useful bound, assume $I_{\sigma\rightarrow\rho}>0$. We define $I_\infty (\rho)$ as the nonuniformity of formation of $\rho$ under NO\cite{dahlsten2011inadequacy}, given by $I_\infty (\rho)=-\log{\eta_1n}$, and hence let $s_{I_{\infty}\left(\rho\right)}$ be the least sharp state that majorizes $\rho$ (see Figure \ref{Lorenz}). Note that $I_{\infty}$ decreases under Noisy Operations and is additive across tensor products \cite{gour2015resource}. In terms of the eigenvalues of $\rho$ and $\sigma$:
\begin{align}
\begin{split}
s_{I_{\infty}\left(\rho\right)}&=s_{\log\left(\eta_1 n\right)},\\
s_{I_{\infty}\left(\sigma\right)}&=s_{\log\left(\zeta_1 n\right)}.
\end{split}
\end{align}
By definition, as $I_{\sigma\rightarrow\rho}>0$:
\begin{equation}
\sigma\stackrel{\textrm{NO}}{\longrightarrow}\rho\otimes s_{I_{\sigma\rightarrow\rho}}.
\end{equation}
Now, using first the monotonicity of $I_{\infty}$ and then the additivity:
\begin{align*}
I_{\infty}\left(\sigma\right)&\geq I_{\infty}\left(\rho\otimes s_{I_{\sigma\rightarrow\rho}}\right),\quad\textrm{(monotonicity)}\\
&=I_{\infty}\left(\rho\right)+I_{\sigma\rightarrow\rho}.\quad\textrm{(additivity)}\\
\Rightarrow I_{\sigma\rightarrow\rho}&\leq I_{\infty}\left(\sigma\right)-I_{\infty}\left(\rho\right),\\
&=\log\left(\zeta_1 n\right)-\log\left(\eta_1 n\right),\\
&=\log\left(\frac{\zeta_1}{\eta_1}\right).\\
\Rightarrow 2^{-I_{\sigma\rightarrow\rho}}&\geq\frac{\eta_1}{\zeta_1},\\
&=\frac{V_1\left(\rho\right)}{V_1\left(\sigma\right)},\\
&\geq p^{*}, \quad\textrm{(by definition)}
\end{align*}
as required.
\end{proof}

From Eq. \eqref{bounds1} we can see that when $I_{\rho\rightarrow\sigma}=-I_{\sigma\rightarrow\rho}\equiv I$ (that is, in a reversible transition) then $p^*=2^{-I}$. This occurs when either $\sigma\stackrel{\text{NO}}{=}\rho \otimes s_{|I|}$ or $\rho\stackrel{\text{NO}}{=}\sigma \otimes s_{|I|}$ depending on whether $I$ is positive or negative (when $I\geq0$ the transition is deterministic).
In terms of Lorenz curves, this means that the curves of $\rho$ and $\sigma$ have the same shape up to re-scaling by a factor $2^{-I}$. 
In particular, this is the case when both $\rho$ and $\sigma$ are sharp states, where both Lorenz curves are straight lines.

This result can be applied in the thermodynamic regime of many independent copies. 
If we want to perform a transition such as: 
\begin{equation}
\rho^{\otimes N} \rightarrow \sigma^{\otimes N},
\end{equation}
we need an amount of work given by $-N I_{\rho\rightarrow\sigma}$. Hence, the probability of success in such a case is bound by:
\begin{equation}
2^{N I_{\rho\rightarrow\sigma}}\leq p^{*} \leq 2^{-N I_{\sigma\rightarrow\rho}},
\end{equation}
which tends to $0$ for large $N$. This can be seen as a way in which in the thermodynamic limit statistical fluctuations are suppressed.

\subsubsection{Lorenz curve interpretation}

In terms of Lorenz curves, adding $I_{\rho \rightarrow \sigma}$ nonuniformity to $\rho$ to make the transition possible is equivalent to compressing the Lorenz curve with respect to the $x$-axis by a ratio $2^{-I_{\rho \rightarrow \sigma}}$, such that the curve of $\rho$ lies just above and touches that of $\sigma$. Hence, a compression by $p^*\ge 2^{-I_{\rho \rightarrow \sigma}}$ must mean that there is at least a point of the compressed curve just below or touching $\sigma$. A proof of this is given in Figure \ref{lowerbound1}. 
\begin{figure}
\includegraphics[width=0.46\textwidth]{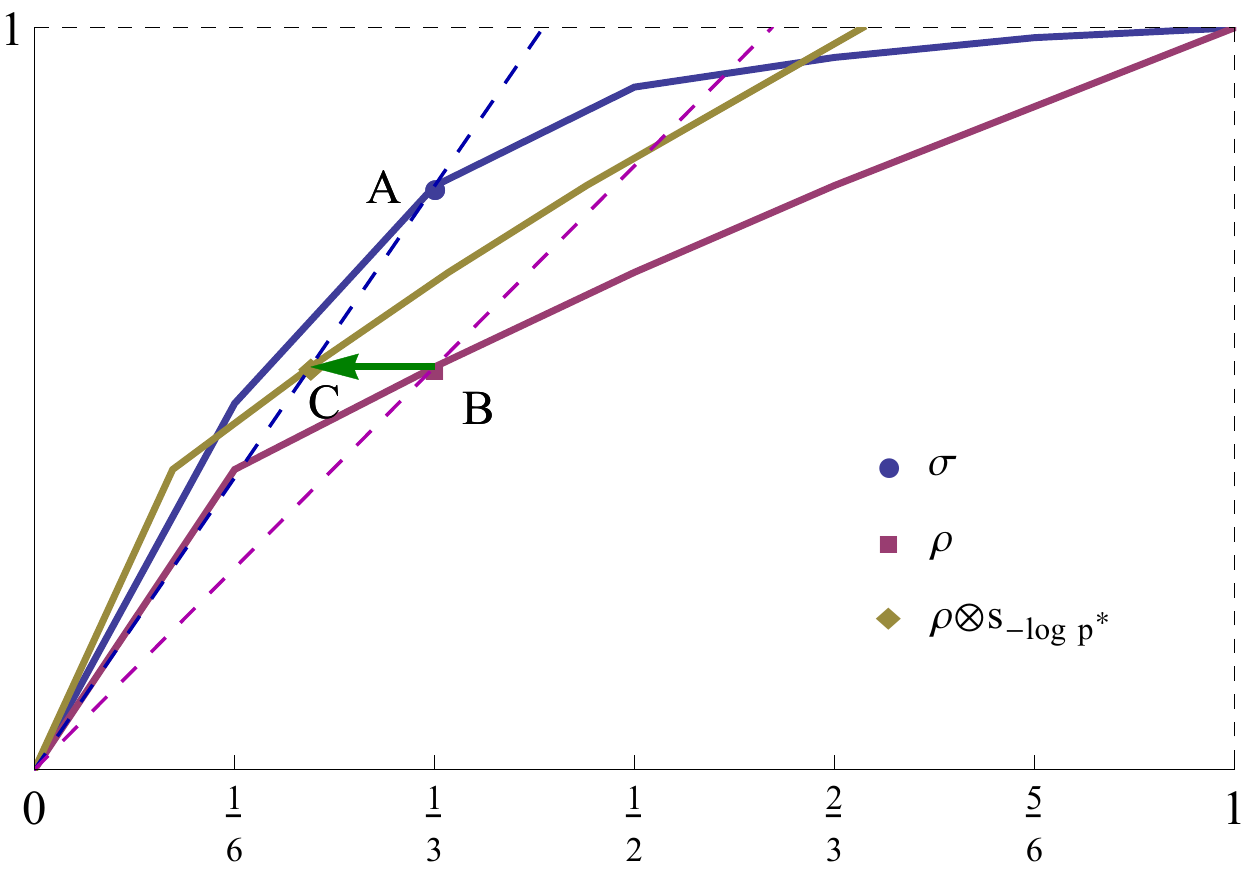}
\caption{We plot the curves of $\rho$, $\sigma$ and $\rho$ compressed by $p^*$ (with respect to the $x$-axis). The points $A$ and $B$ at which the vertical ratio between the curves of $\rho$ and $\sigma$ is maximum (which sets $l_1$ and $p^*$), and the sharp states that pass through those points are also shown as dashed lines. After compressing the Lorenz curve of $\rho$ by a ratio of $p^*$, the point $B$ will be taken to $C$, which will always either be below the curve of $\sigma$ or just touching it. This proves the lower bound in Eq. \eqref{bounds1}.} \label{lowerbound1}
\end{figure}

Extracting $I_{\sigma \rightarrow \rho}$ nonuniformity from $\sigma$ before performing NO into $\rho$ is equivalent to compressing the curve of $\rho$ by a ratio of $2^{-I_{\sigma \rightarrow \rho}}$ such that the curve of $\sigma$ lies just above and touches that of $\rho$. Hence, to prove the upper bound in Eq. \eqref{bounds1}, it suffices to show that in compressing the curve of $\rho$ by $p^*$ at least one point of the new curve must lie above or touch that of $\sigma$. In Figure \ref{upperbound1} we show a diagrammatic version of the proof given in Section \ref{work}.

\begin{figure}
\includegraphics[width=0.46\textwidth]{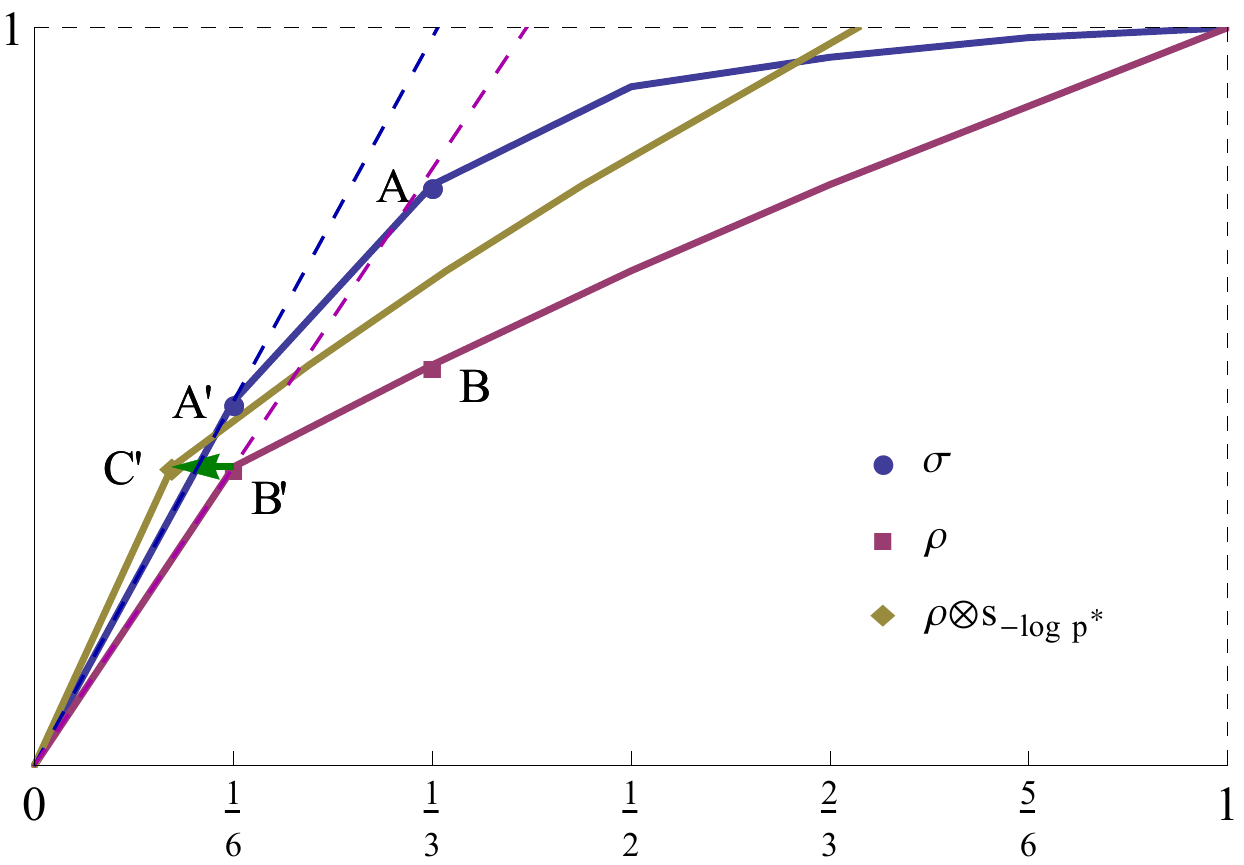}
\caption{We plot the curves of $\rho$, $\sigma$ and $\rho$ compressed by $p^*$ (with respect to the $x$-axis). The points $A$ and $B$ at which the vertical ratio between the curves of $\rho$ and $\sigma$ is maximum (which sets $l_1$ and $p^*$) and the sharp states $I_{\infty}\left(\rho\right)$ and $I_{\infty}\left(\sigma\right)$ are also shown as dashed lines. Given that for sharp states all bounds are saturated, the appropriate maximum vertical and horizontal ratios coincide, and are $\eta_1/\zeta_1$, the ratio of the heights of $B'$ and $A'$. But this ratio is, by definition, bigger than or equal to $p^*$, the ratio between $A$ and $B$. This means that if the curve of $\rho$ is compressed by $p^*$, the point $B'$ is mapped to $C$ just above or touching the curve of  $\sigma$, proving the upper bound of Eq. \eqref{bounds1}.}\label{upperbound1}
\end{figure}

It should be noted that with Lemma \ref{bounds} we are proving a general statement about convex Lorenz curves. This is, that the minimum vertical ratio of two given curves ($p^*$) is lower and upper bounded respectively by the minimum and the maximum horizontal ratio of the two.

\section{Probability of transition under Thermal Operations}
\label{sec:TO}

Noisy Operations can be generalized to include systems with arbitrary, finite Hamiltonians. This is the resource theory of Thermal Operations \cite{Streater_dynamics,janzing2000thermodynamic,HO-limitations,brandao2013resource}. Within this scheme, the allowed operations are: \emph{i)} a system with any Hamiltonian in the Gibbs state of that Hamiltonian can be added, \emph{ii)} any subsystem can be discarded through tracing out and \emph{iii)} any energy-conserving unitary, i.e. those unitaries that commute with the total Hamiltonian, can be applied to the global system. These operations model the thermodynamics of a system in the presence of an ideal heat bath \cite{HO-limitations,brandao2013resource}. Note that while the heat bath the system is in contact with is assumed to be large, thermal operations include processes that only interact with a small part of the bath. As such, limitations derived with respect to such an idealized bath can be regarded as truly fundamental. Even though the bath size can be large, the system of interest is fixed, and can for example, be only a single system. They thus describe processes beyond the thermodynamic limit.

In general, the initial and final systems may have different Hamiltonians but, by making use of the `switching qubit' construction in \cite{HO-limitations}, we can w.l.o.g. assume that the initial and final Hamiltonians are the same. As such, the results in this section will assume this but in Section \ref{sec:Ham Change} we will discuss how a changing Hamiltonian affects them. In Appendix H of \cite{brandao2013resource} it was shown that other mainstream thermodynamical paradigms such as time dependent Hamiltonians, the insertion of interaction terms between system, bath and work systems and various master equations are all included within the scope of Thermal Operations.

In the absence of catalysts, and provided the final state is block-diagonal in the energy eigenbasis, it was established in \cite{HO-limitations} that a transition from $\rho$ to $\sigma$ is possible under Thermal Operations if and only if $\rho$ thermo-majorizes $\sigma$. This is similar in form to the majorization criteria of Noisy Operations and can be visualized in terms of \emph{thermo-majorization diagrams} which are similar to Lorenz curves but with two crucial differences.

Suppose $\rho$ is also block-diagonal in the energy eigenbasis with eigenvalue $\eta_i$ associated with energy level $E_i$, for $1\leq i\leq n$. Firstly, rather than ordering according to the magnitude of $\eta_i$, we instead \emph{$\beta$-order} them, listing $\eta_i e^{\beta E_i}$ in descending order.

The second difference is that we no longer plot the $\beta$-ordered $\eta_i$ at evenly spaced intervals. Instead we plot the points:
\begin{equation}
\left\{\left(\sum_{i=1}^{k}e^{-\beta E_{i}^{\left(\rho\right)}},\sum_{i=1}^{k}\eta_i^{\left(\rho\right)}\right)\right\}_{k=1}^{n},
\end{equation} 
where the superscript $\rho$ on $E_i$ and $\eta_i$ indicates that they have been $\beta$-ordered and this ordering depends on $\rho$. Thermo-majorization states that $\rho$ can be deterministically converted into a block-diagonal $\sigma$ if and only if its thermo-majorization curve never lies below that of $\sigma$, as is shown in Figure \ref{Thermomaj}. This is analogous to the case of Noisy Operations. In what follows, we assume that the $\eta_i$ have been $\beta$-ordered unless otherwise stated.

If $\rho$ is not block-diagonal in the energy eigenbasis, to determine if a transition is possible we consider the thermo-majorization curve associated with the state formed by decohering $\rho$ in the energy eigenbasis. This state, $\rho_D$, is given by:
\begin{equation} \label{eq:decohere}
\rho_D=\sum_{i=1}^{n}\ketbra{E_i}{E_i}\rho\ketbra{E_i}{E_i},
\end{equation}
where $\ket{E_i}$ is the eigenvector of the system's Hamiltonian associated with energy level $E_i$. The operation of decohering $\rho$ to give $\rho_D$ is a Thermal Operation and commutes with all other Thermal Operations \cite{brandao2013resource}. A transition from $\rho$ to $\sigma$, where $\sigma$ is block-diagonal in the energy eigenbasis, can be made deterministically if and only if the thermo-majorization curve of $\rho_D$ is never below that of $\sigma$.

Finally, if $\sigma$ is not block-diagonal, a transition from $\rho$ to $\sigma$ is possible only if $\rho_D$ thermo-majorizes $\sigma_D$ and finding a set of sufficient conditions is an open question.


In what follows, the thermo-majorization curve of a state with coherences is defined to be the thermo-majorization curve of that state decohered in the energy eigenbasis as per Eq. \eqref{eq:decohere}.

\begin{figure}
\includegraphics[width=0.46\textwidth]{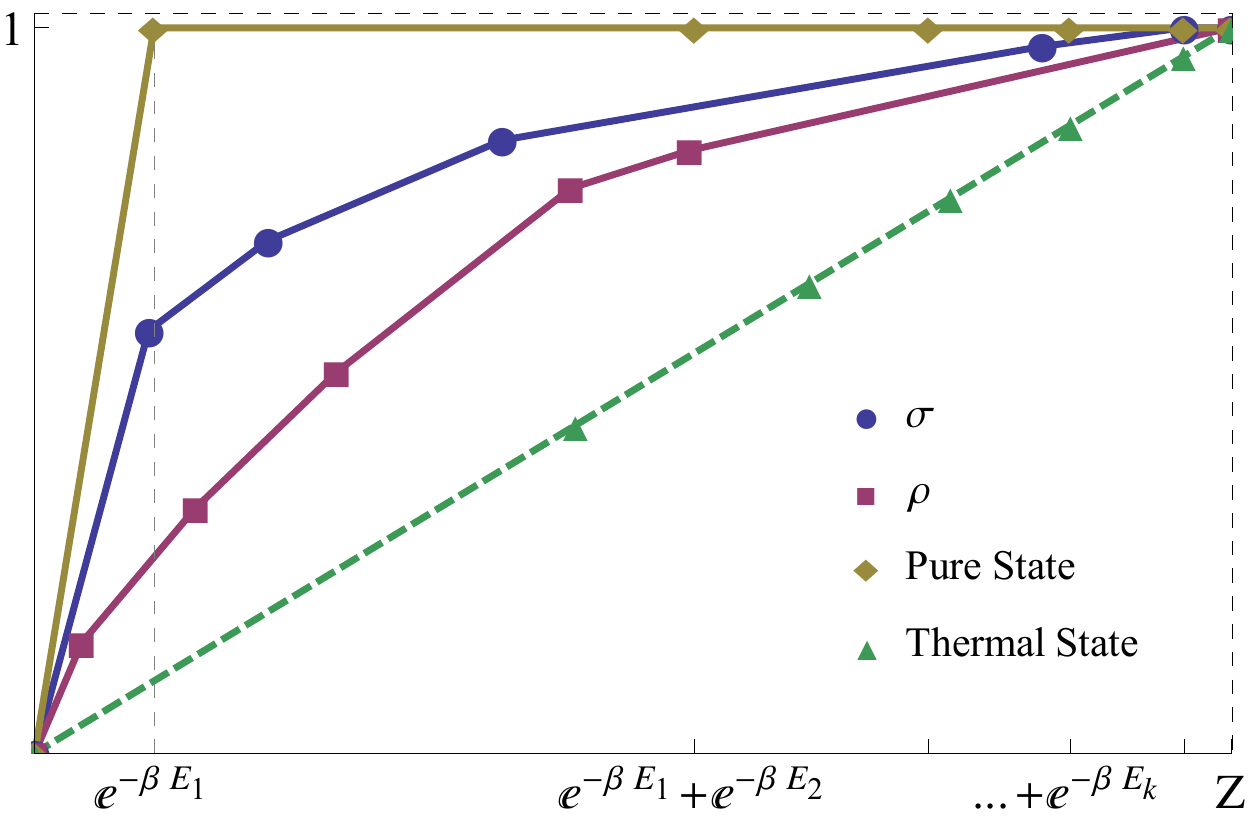}
\caption  {We show the $\beta$-ordered thermo-majorization diagrams for various states of the system. Note that different states may have different $\beta$-orderings and the markings on the $x$-axis correspond to one particular $\beta$-ordering. The curves always end at $(Z,1)$. The thermo-majorization criteria states that we can take a state to another under Thermal Operations if and only if the curve of the initial state is above that of the final state. Hence, in this case (provide $\rho$ is block-diagonal in the energy eigenbasis) there is a set of operations such that $\sigma \stackrel{TO}{\rightarrow} \rho$, but not for the reverse process.} \label{Thermomaj}
\end{figure}

Similarly to how Eq. \eqref{monotones} defines monotones for the Noisy Operations resource theory, the height of the $\beta$-ordered thermo-majorization curves provides monotones for Thermal Operations. If we denote the height of the thermo-majorization curve of $\rho$ at $x$ by $\tilde{V}_x\left(\rho\right)$, for $0\leq x\leq Z$ (where $Z$ is the partition function), then by the thermo-majorization criteria, this function is non-increasing under Thermal Operations. In particular, for block-diagonal $\rho$, we have:
\begin{equation}
\tilde{V}_{x_k}\left(\rho\right)=\sum_{i=1}^{k}\eta_i^{\left(\rho\right)}, \quad {\textrm{where } x_k=\sum_{i=1}^{k}e^{-\beta E_{i}^{\left(\rho\right)}}}.
\label{eq:TOmonotones}
\end{equation}
These monotones also give us an alternative way of stating the thermo-majorization criteria:
\begin{theorem} \label{Finite set theorem}
Suppose $\sigma$ is block-diagonal in the energy eigenbasis.
Let $\mathcal{L}\left(\sigma\right)=\left\{\sum_{i=1}^{k}e^{-\beta E_{i}^{\left(\sigma\right)}}\right\}_{k=1}^{n}$. Then $\rho$ can be deterministically converted into $\sigma$ under Thermal Operations if and only if:
\begin{equation} \label{Discrete thermo majorization}
\tilde{V}_x\left(\rho\right)\geq \tilde{V}_x\left(\sigma\right), \quad\forall x\in\mathcal{L}\left(\sigma\right).
\end{equation}
\end{theorem}
\begin{proof}
To prove this theorem, we make use of the concavity properties of thermo-majorization curves.
Suppose $\rho\stackrel{TO}{\longrightarrow}\sigma$. Then by thermo-majorization, $\tilde{V}_x\left(\rho\right)\geq\tilde{V}_x\left(\sigma\right)$, for $0\leq x\leq Z$ and in particular Eq. \eqref{Discrete thermo majorization} holds.

Conversely, suppose Eq. \eqref{Discrete thermo majorization} holds and, setting $t_0=0$, label the elements of $\mathcal{L}\left(\sigma\right)$ arranged in increasing order by $t_i$ for $i=1$ to $n$. Then on the interval $\left[t_{i-1},t_{i}\right]$, for $1\leq i\leq n$, the thermo-majorization curve of $\sigma$ is given by a straight line. From $\rho$, define the block-diagonal state $\rho_{\sigma}$ by the thermo-majorization curve:
\begin{equation} \label{rho_sigma}
\left\{\left(t_i,\tilde{V}_{t_i}\left(\rho\right)\right)\right\}_{i=1}^{n},
\end{equation}
and note that due to the concavity of thermo-majorization curves, $\rho$ thermo-majorizes $\rho_{\sigma}$. On the interval $\left[t_{i-1},t_{i}\right]$, $1\leq i\leq n$, the thermo-majorization curve of $\rho_{\sigma}$ is also given by a straight line. The construction of this state $\rho_\sigma$ is shown in Figure \ref{Theorem5-2}.

\begin{figure}
\includegraphics[width=0.46\textwidth]{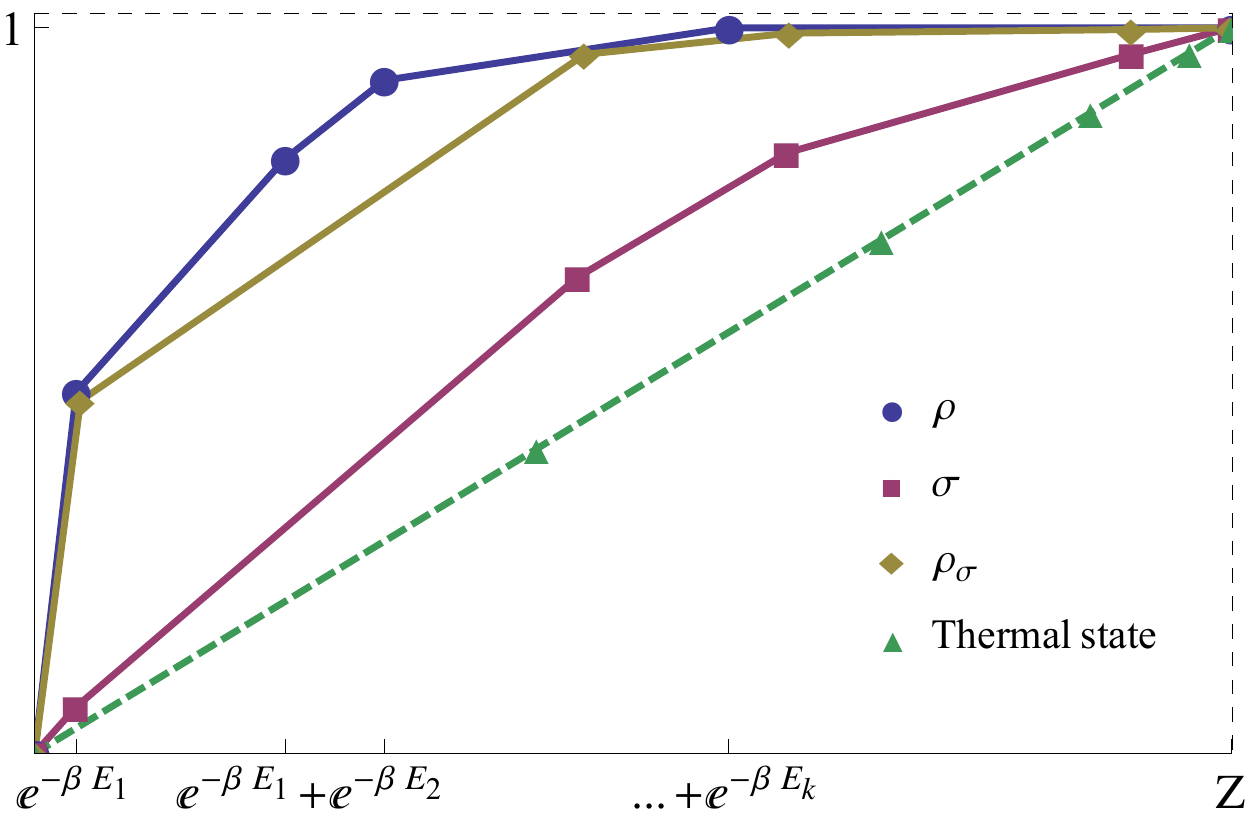}
\caption{Here we illustrate the construction of the state $\rho_\sigma$ used in the proof of Theorem \ref{Finite set theorem}. The points of the curve $\rho$ that are at the same horizontal position as the elbows of $\sigma$ are joined, and by concavity the resultant curve is always below $\rho$.} \label{Theorem5-2}
\end{figure}

As $\tilde{V}_{t_i}\left(\rho_{\sigma}\right)=\tilde{V}_{t_i}\left(\rho\right)$, $\forall i$ by construction, Eq. \eqref{Discrete thermo majorization} implies that $\tilde{V}_{t_i}\left(\rho_{\sigma}\right)\geq \tilde{V}_{t_i}\left(\sigma\right)$, $\forall i$. Hence on the interval $\left[t_{i-1},t_{i}\right]$, $1\leq i\leq n$, the thermo-majorization curves for $\rho_\sigma$ and $\sigma$, and therefore $\rho$ and $\sigma$, do not cross. As this holds for all $i$ and the intervals cover $[0,Z]$ the thermo-majorization curve of $\rho$ is never below that of $\sigma$ and we can perform $\rho\stackrel{TO}{\longrightarrow}\sigma$ deterministically.
\end{proof}
If we define the number of `elbows' in the thermo-majorization curve of $\sigma$ to be $j$, this reduces thermo-majorization to checking $j$ criteria and generalizes Lemma 17 of \cite{gour2015resource} to Thermal Operations. Note also that if $\sigma$ is not block-diagonal in the energy eigenbasis,  Eq. \eqref{Discrete thermo majorization} gives a necessary but not sufficient condition for the transition from $\rho$ to $\sigma$ to be possible.

\subsection{Non-deterministic transformations} \label{ssec:non-det TO}

Having defined the appropriate monotones for Thermal Operations, we are now in a position to investigate non-deterministic transformations and prove a theorem analogous to Theorem \ref{NO Theorem}.

\begin{theorem}\label{pthermal}
Suppose we wish to transform the state $\rho$ to the state $\sigma$ under Thermal Operations. The maximum value of $p$, $p^*$, that can be achieved in the transition:
\begin{equation} \label{TO Goal}
\rho\stackrel{\textit{TO}}{\longrightarrow}\rho'= p\sigma+\left(1-p\right)X,
\end{equation}
is such that:
\begin{equation} \label{TO p*}
p^{*}\leq\min_{x\in\mathcal{L}\left(\sigma\right)} \frac{\tilde{V}_x\left(\rho\right)}{\tilde{V}_x\left(\sigma\right)}.
\end{equation}
Furthermore, if $\sigma$ is block-diagonal in the energy eigenbasis, there exists a protocol that achieves the bound.
\end{theorem}
\begin{proof}
Proving this result is more complicated than proving Theorem \ref{NO Theorem} due to the fact that $\rho$ and $\sigma$ may have different $\beta$-orderings. We proceed as before, first showing the bound in Eq. \eqref{TO p*} and then giving a protocol that achieves the bound when $\sigma$ is block-diagonal.

We prove the bound in Eq. \eqref{TO p*} by constructing useful intermediate curves between those of $\rho$ and $p \sigma$ to deal with differing $\beta$-orders. With these in place, the result will follow in a similar manner to Theorem \ref{NO Theorem}.

We begin by showing that given Eq. \eqref{TO Goal}:
\begin{equation}
\tilde{V}_x\left(\rho\right)\geq p \tilde{V}_x\left(\sigma\right), \quad \forall x \in [0,Z].
\end{equation}
First consider (for general $\sigma$) the maximum value of $p$ that can be achieved in attempting to convert $\rho$ into $\sigma$. As decohering is a Thermal Operation, this value of $p$ can also be achieved when attempting to convert $\rho$ into $\sigma_D$:
\begin{align*}
\rho&\stackrel{\textit{TO}}{\phantom{A}\longrightarrow\phantom{A}}\rho'_{\phantom{D}}= p\sigma+\left(1-p\right)X,\\
&\stackrel{\textit{decohere}}{\phantom{A}\longrightarrow\phantom{A}}\rho'_D= p\sigma_D+\left(1-p\right)X_D.
\end{align*}
Thus, to upper bound $p^*$, it suffices to show that Eq. \eqref{TO p*} holds for block-diagonal $\sigma$. Furthermore, w.l.o.g. we can assume that $\rho'$ and $X$ are also block-diagonal. Using Weyl's inequality as per Theorem \ref{NO Theorem} to deal with degenerate energy levels, for block-diagonal $\rho'$, $\sigma$ and $X$, we have:
\begin{equation} \label{TO evalues}
{\eta}_i'\geq p\zeta_i, \quad\forall i. 
\end{equation}

Now consider the sub-normalized thermo-majorization curve of $p \sigma$ given by the points:
\begin{equation} \label{Lorenz sigma}
\left\{\left(\sum_{i=1}^k e^{-\beta E^{\left(\sigma\right)}_i},p\sum_{i=1}^{k}\zeta_i^{\left(\sigma\right)}\right)\right\}_{k=1}^{n},
\end{equation}
and the (possibly non-concave) curve formed by plotting the eigenvalues of $\rho'$ according to the $\beta$-ordering of $\sigma$. This is given by the points:
\begin{equation} \label{Lorenz rho'sigma}
\left\{\left(\sum_{i=1}^k e^{-\beta E^{\left(\sigma\right)}_i},\sum_{i=1}^{k}{\eta'_i}^{\left(\sigma\right)}\right)\right\}_{k=1}^{n}.
\end{equation}
By Eq. \eqref{TO evalues}, the curve defined in Eq. \eqref{Lorenz rho'sigma} is never below that defined in Eq. \eqref{Lorenz sigma}.

Finally, the thermo-majorization curve of $\rho'$ is given by:
\begin{equation} \label{Lorenz rho'}
\left\{\left(\sum_{i=1}^k e^{-\beta E^{\left(\rho'\right)}_i},\sum_{i=1}^{k}{\eta'_i}^{\left(\rho'\right)}\right)\right\}_{k=1}^{n}.
\end{equation}
Note that attempting to construct a thermo-majorization curve for $\rho'$ with respect to the $\beta$-ordering of another state, as we do in Eq. \eqref{Lorenz rho'sigma}, has the effect of rearranging the piecewise linear segments of the true thermo-majorization curve. This means that they may no longer be joined from left to right in order of decreasing gradient. Such a curve will always be below the true thermo-majorization curve. 
To see this, imagine constructing a curve from the piecewise linear elements and in particular, trying to construct a curve that would lie above all other possible constructions. Starting at the origin, we are forced to choose the element with the steepest gradient - all other choices would lie below this by virtue of having a shallower gradient. We then proceed iteratively, starting from the endpoint of the previous section added and choosing the element with the largest gradient from the remaining linear segments. The construction that we obtain is the true thermo-majorization curve. A graphical description of this proof is shown in Figure \ref{Theorem5}.

\begin{figure}
\includegraphics[width=0.46\textwidth]{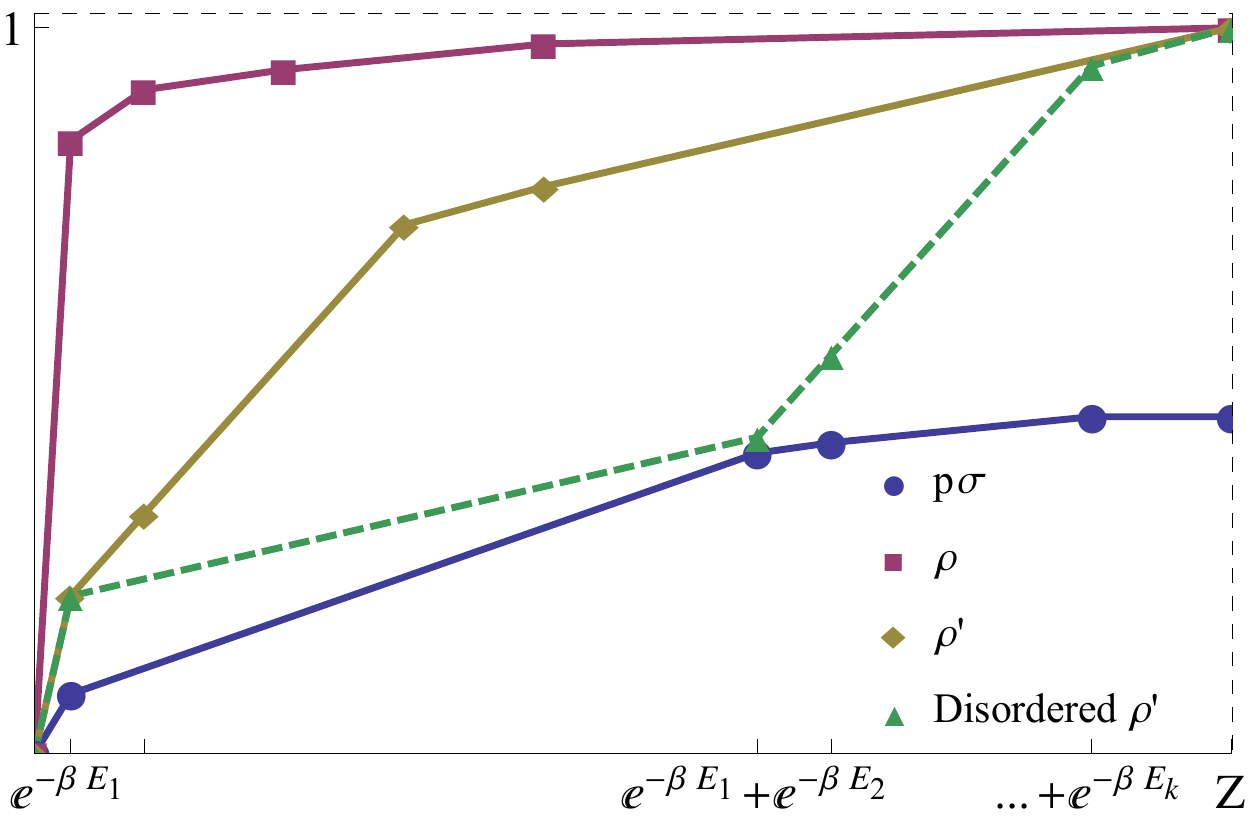}
\caption{Here we show graphically the steps of the proof of the first part of Theorem \ref{pthermal}. In the decomposition of Eq. \eqref{TO Goal} the curve $p \sigma$ must always be below that of $\rho'$ and hence also $\rho$. This sets the maximum probability $p^*$ as defined in Eq. \eqref{TO p*}. Both $p \sigma$ and the disordered $\rho'$ have the same $\beta$-ordering.} \label{Theorem5}
\end{figure}

As such, the curve in Eq. \eqref{Lorenz rho'} is never below that in Eq. \eqref{Lorenz rho'sigma}. This gives us:
\begin{equation}
\tilde{V}_x\left(\rho\right)
\geq \tilde{V}_x\left(\rho'\right)
\geq p\tilde{V}_x\left(\sigma\right),
\end{equation}
where the first inequality holds as, by definition, $\rho$ thermo-majorizes $\rho'$. In particular we have:
\begin{equation}
p^{*}\leq\min_{x\in\mathcal{L}\left(\sigma\right)} \frac{\tilde{V}_x\left(\rho\right)}{\tilde{V}_x\left(\sigma\right)}.
\end{equation}

When $\sigma$ is block-diagonal in the energy eigenbasis, a protocol that saturates the bound is:
\begin{align*}
\rho&\stackrel{\textit{TO}}{\longrightarrow}\rho_{\sigma},\\
&\stackrel{\textit{TO}}{\longrightarrow}\rho_{\sigma}'=p^* \sigma +\left(1-p^*\right) X,
\end{align*}
where $\rho_{\sigma}$ was defined in Eq. \eqref{rho_sigma} and is thermo-majorized by $\rho$. As $\rho_{\sigma}$ and $\sigma$ have the same $\beta$-ordering and:
\begin{equation}
\frac{\tilde{V}_x\left(\rho\right)}{\tilde{V}_x\left(\sigma\right)}=\frac{\tilde{V}_x\left(\rho_{\sigma}\right)}{\tilde{V}_x\left(\sigma\right)}, \quad\forall x\in\mathcal{L}\left(\sigma\right),
\end{equation}
applying the same construction used in Theorem \ref{NO Theorem} gives a strategy to produce $\rho'_{\sigma}$ that achieves:
\begin{equation}
p^*=\min_{x\in\mathcal{L}\left(\sigma\right)} \frac{\tilde{V}_x\left(\rho\right)}{\tilde{V}_x\left(\sigma\right)}.
\end{equation}
\end{proof}

\subsection{Measuring whether the transition occurred under Thermal Operations}
\label{sec:measurement}

For block-diagonal $\sigma$, after obtaining $\rho'$ through Thermal Operations we may apply the measurement defined by Eq. \eqref{POVM} to extract our target state with probability $p^*$. This can be done through a process that uses an ancilla qubit system, $Q$, that starts and ends in the state $\ket{0}$ and has associated Hamiltonian, $H_Q=\mathbb{I}_2$, a unitary that correlates the system with the ancilla and a projective measurement on the ancilla qubit. As the measurement operators are diagonal in the energy eigenbasis, we will find that the unitary is energy conserving and within the set of Thermal Operations. Furthermore, the ancilla that is used to perform the POVM 
can be returned back into it's original state. Hence the only cost we have to pay is to erase the record of the measurement outcome itself. As is well known \cite{landauer}, the cost of erasing the record is $kT\log{2}$, although if one is repeating the process many times, then it is $kTh\left(p^*\right)$ with $h\left(p^*\right)$ the binary entropy $h\left(p^*\right)=-p^*\log{p^*}-\left(1-p^*\right)\log{\left(1-p^*\right)}$ \cite{skrzypczyk2014work}.

The unitary that we shall use is given by:
\begin{align}\label{unitary}
U_{\text{SQ}}=
\begin{pmatrix}
\sqrt{M} & \sqrt{\mathbb{I}-M} \\
\sqrt{\mathbb{I}-M} & -\sqrt{M}
\end{pmatrix},
\end{align}
where $M$ is defined as per Eq. \eqref{POVM}.
Note that $U_{\text{SQ}}=U^{\dagger}_{\text{SQ}}$. Its effect on the initial joint state is:
\begin{align*}
&U_{\text{SQ}}(\rho' \otimes \ketbra{0}{0})U^{\dagger}_{\text{SQ}},  \\
=& \begin{pmatrix}
\sqrt{M} & \sqrt{\mathbb{I}-M} \\
\sqrt{\mathbb{I}-M} & -\sqrt{M}
\end{pmatrix}
\begin{pmatrix}
\rho' & 0 \\
0 & 0
\end{pmatrix}
 \begin{pmatrix}
\sqrt{M} & \sqrt{\mathbb{I}-M} \\
\sqrt{\mathbb{I}-M} & -\sqrt{M}
\end{pmatrix},\\ 
=&  \begin{pmatrix}
\sqrt{M} \rho' \sqrt{M} & \sqrt{M} \rho' \sqrt{\mathbb{I}-M} \\
\sqrt{\mathbb{I}-M}\rho' \sqrt{M} & \sqrt{\mathbb{I}-M}\rho' \sqrt{\mathbb{I}-M}
\end{pmatrix},
\\ 
=&  \begin{pmatrix}
p^* \sigma & \sqrt{M} \rho' \sqrt{\mathbb{I}-M} \\
\sqrt{\mathbb{I}-M}\rho' \sqrt{M} & (1-p^*) X
\end{pmatrix}.
\end{align*}
If we now measure the ancilla in the computational basis, the joint state will collapse to $\sigma\otimes\ketbra{0}{0}$ when the 0 outcome is observed. This happens with probability $p^*$. If the 1 outcome is observed, the joint state collapses to $X\otimes\ketbra{1}{1}$ and this happens with probability $1-p^*$. In addition, if the 1 outcome is observed, we can then apply a Pauli $Z$ to the ancilla qubit to return it to its initial state.

To see that $U_{SQ}$ commutes with the total Hamiltonian and belongs to the class of Thermal Operations, first note that the total Hamiltonian is given by:
\begin{equation}
H_{SQ}=H_S\otimes\mathbb{I}_2 + \mathbb{I}_n \otimes\mathbb{I}_2.
\end{equation}
The unitary trivially commutes with the second term so focusing on the first term, and noting that $M$ and $H_S$ are both diagonal matrices so commute, it is easy to check that:
\begin{align*}
[U_{\text{SQ}},H_S\otimes\mathbb{I}_2]&=
\begin{pmatrix}
\sqrt{M} & \sqrt{\mathbb{I}-M} \\
\sqrt{\mathbb{I}-M} & -\sqrt{M}
\end{pmatrix}
\begin{pmatrix}
H_{\text{S}} & 0 \\
0 & H_{\text{S}}
\end{pmatrix}
\\
&\quad- \begin{pmatrix}
H_{\text{S}} & 0 \\
0 & H_{\text{S}}
\end{pmatrix}
\begin{pmatrix}
\sqrt{M} & \sqrt{\mathbb{I}-M} \\
\sqrt{\mathbb{I}-M} & -\sqrt{M}
\end{pmatrix},\\
&=0.
\end{align*}
Hence $\left[U_{SQ},H_{SQ}\right]=0$.

Observe that this reasoning can be generalized to measurements with $s$ outcomes \cite{navascues2014energy}. Provided the measurement operators commute with $H_S$, the measurement can be performed using a $s$-level ancilla system with trivial Hamiltonian and a joint energy-conserving unitary. Such a measurement can be performed for free up to having to spend work to erase the record of the measurement outcome at a cost of $kT\ln s$. On the other hand, channels that are not composed of Thermal Operations (including some measurements characterized by non-diagonal operators) can be seen as a resource \cite{navascues2015nonthermal}.

\subsection{Work of transition under Thermal Operations} \label{sec: TO work of trans}

\subsubsection{Work systems}

In general, if we want a transition $\rho \rightarrow \sigma$ to be possible, work may have to be supplied. Alternatively, if a transition can be achieved with certainty, it may be possible to extract work. For the thermodynamics of small systems, the concept of \emph{deterministic work} (also referred to in the literature as single-shot or worst-case work) has been introduced \cite{dahlsten2011inadequacy,HO-limitations, aaberg-singleshot}.


Within the Thermal Operation paradigm, the optimal amount of work that must be added or gained can be quantified using the energy gap, $W$, of a 2-level system with ground state $\ket{0}$ and excited state $\ket{W}$ with energy $W$. The associated Hamiltonian is:
\begin{equation}
H=W\ketbra{W}{W}.
\end{equation}
The work of transition, $W_{\rho\rightarrow\sigma}$, is such that:
\begin{equation} \label{eq:TO work}
\begin{split}
\textrm{if }W_{\rho\rightarrow\sigma} \le 0,\quad\quad\quad\quad\quad&\\
\rho\otimes \ketbra{W_{\rho\rightarrow\sigma}}{W_{\rho\rightarrow\sigma}}&\stackrel{TO}{\longrightarrow}\sigma\otimes \ketbra{0}{0},\\
\textrm{if }W_{\rho\rightarrow\sigma}>0,\quad\quad\quad\quad\quad&\\
\rho\otimes \ketbra{0}{0}&\stackrel{TO}{\longrightarrow}\sigma \otimes \ketbra{W_{\rho\rightarrow\sigma}}{W_{\rho\rightarrow\sigma}}.
\end{split}
\end{equation}

Defining work in such a way enables the quantification of the worst-case work of a process. When $W_{\rho\rightarrow\sigma}$ is negative, it can be interpreted as the smallest amount of work that must be supplied to guarantee the transition takes place. If it is positive, it is the largest amount of work we are guaranteed to extract in the process. As the work system is both initially and finally in a pure state, no entropy is contained within it and its energy change must be completely due to work being exchanged with the system. Given the energy-conservation law that Thermal Operations follow (equivalent to the first law), this idea of work automatically yields a definition of what heat is. In a given operation, the change in energy of work bit, system and heat bath must be zero, and hence we can straightforwardly identify heat as the change in energy of the heat bath, or minus the change in energy on system and work bit.

As we illustrate in Figure \ref{WorkBit}, the effect of appending a pure state of work to $\rho$ is equivalent to stretching the thermo-majorization curve by a factor of $e^{-\beta W}$, and tensoring by the corresponding ground state to $\sigma$ does not change the curve \cite{HO-limitations}. In both cases the $\beta$-order is preserved, and the new curves will have a lengthened $x$-axis $\left[0,Z\left(1+e^{-\beta W}\right)\right]$. These different stretchings can serve to place the curve of $\rho$ just above that of $\sigma$, in which case $W$ will be the work of transition, in a similar way to the case of nonuniformity within Noisy Operations. 

\begin{figure}
\includegraphics[width=0.46\textwidth]{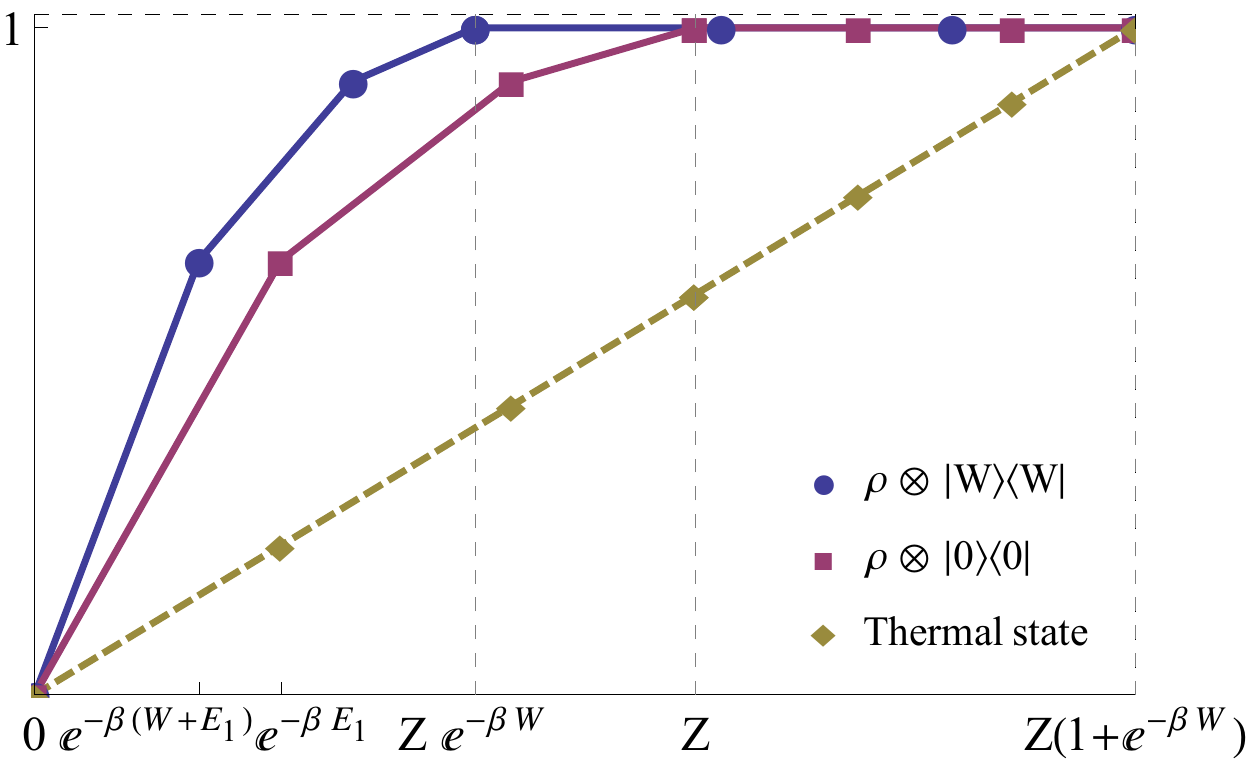}
\caption{We show the thermo-majorization curves of a state to which a work qubit in one of two pure states has been tensored. Adding this work system takes $Z \rightarrow Z\left(1+e^{-\beta W}\right)$, extending the $x$-axis. When we tensor with the ground state to form $\rho \otimes \ket{0}\bra{0}$, the curve is the same as for $\rho$ alone, but when the excited state is tensored, there is a change in the energy levels of the $\beta$-ordering, and as a result the curve of $\rho$ is compressed by a ratio of $e^{-\beta W}$.} \label{WorkBit}
\end{figure}


\subsubsection{Monotones under Thermal Operations, and the work of transition}
\label{ss:work}

In Thermal Operations, the horizontal distance between a state's thermo-majorization curve and the $y$-axis is again a monotone for each value of $y\in\left[0,1\right]$. We denote these by $\tilde{L}_y$ and, as before, they never decrease under Thermal Operations. In particular, for block-diagonal $\rho$, we have:
\begin{align}
\begin{split}
\tilde{L}_{y_k}\left(\rho\right)&=\sum_{i=1}^{k}e^{-\beta E_{i}^{\left(\rho\right)}}, \textrm{ for } y_k=\sum_{i=1}^{k}\eta_i^{(\rho)}, \quad 1\leq k< \textrm{rank}\left(\rho\right),\\
\tilde{L}_{1}\left(\rho\right)&=\sum_{i=1}^{\textrm{rank}\left(\rho\right)}e^{-\beta E_{i}^{\left(\rho\right)}},
\end{split}
\end{align}
where all sums have been properly $\beta$-ordered.

Similarly to Lemma \ref{NO Work Exp} we have:
\begin{lemma} \label{TO Work Exp}
Given two states $\rho$ and $\sigma$,
where $\sigma$ is block-diagonal in the energy eigenbasis,
 under Thermal Operations:
\begin{equation} \label{eq:TO Work Exp}
e^{-\beta W_{\rho\rightarrow\sigma}}=\max_{y\in\mathcal{D}\left(\sigma\right)}\frac{\tilde{L}_y\left(\rho\right)}{\tilde{L}_y\left(\sigma\right)}.
\end{equation}
\end{lemma}
The proof is near identical to that given in Lemma \ref{NO Work Exp} for Noisy Operations and so we omit it here.

If $\sigma$ is not block-diagonal, the right hand side of Eq. \eqref{eq:TO Work Exp} lower bounds $e^{-\beta W_{\rho\rightarrow\sigma}}$. 
To see this, recall that decohering commutes with Thermal Operations, and hence if the transition $\rho \otimes \ket{0}\bra{0} \rightarrow \sigma \otimes \ket{W_{\rho\rightarrow\sigma}}\bra{W_{\rho\rightarrow\sigma}}$ is possible, so is $\rho \otimes \ket{0}\bra{0} \rightarrow \sigma_D \otimes \ket{W_{\rho\rightarrow\sigma}}\bra{W_{\rho\rightarrow\sigma}}$, and hence $W_{\rho\rightarrow\sigma}\leq W_{\rho\rightarrow\sigma_D}$.

\subsubsection{Bounds on the transition probability}
\label{ss:TObounds}

We can prove a result analogous to Eq. \eqref{bounds1} for the thermal case:

\begin{lemma}\label{thermobounds}
Given two states $\rho$ and $\sigma$,
where $\sigma$ is block-diagonal in the energy eigenbasis,
under Thermal Operations:
\begin{equation}
e^{\beta W_{\rho \rightarrow \sigma}} \le p^* \le e^{-\beta W_{\sigma \rightarrow \rho_D}},
\end{equation}
where as $p^{*}\leq 1$, we assume $W_{\rho\rightarrow\sigma}\leq0$. If $W_{\rho\rightarrow\sigma}\geq0$, $p^*=1$ and the transformation from $\rho$ to $\sigma$ can be done deterministically, potentially extracting a finite amount of work.
\end{lemma}
\begin{proof}
The previous Lemma \ref{bounds} can be seen as a general statement about pairs of concave Lorenz-like curves: the minimum vertical ratio is lower and upper bounded by the minimum and maximum horizontal ratios of the two. Given our previous definitions of the work of transition, and the fact that $p^*$ is the minimum vertical ratio of the two Lorenz curves (as shown in Theorem \ref{pthermal}), the result follows.

\end{proof}

The upper bound of Lemma \ref{thermobounds} can be related to the Jarzynski equality, which is found to hold for general thermal operations applied to the system in an initial thermal state,
(see \cite{alhambra2016second} for further details). The equality states that, for a given thermal operation that extracts work $w$ with some probability $p\left(w\right)$, we have that:
\begin{equation}
\langle e^{\beta w} \rangle = \sum_w e^{\beta w} p(w)=1.
\end{equation}

The Jarzynski equation is valid if the initial state is thermal, so let us take the special case of Lemma \ref{thermobounds}, of a process where we start with a thermal state $\tau$ and probabilistically go to some $\sigma$ diagonal in energy, with optimal probability $p^*$. Because $\tau$ is the fixed point, the effect of that operation is trivial:

\begin{equation} \label{eq:tautosigma}
\tau \stackrel{TO}{\longrightarrow} \rho'=\tau=p^* \sigma + (1-p^*) X.
\end{equation}
Now, if we append an idealized weight with Hamiltonian $H_W=\int_\mathbb{R}dw\, w\ketbra{w}{w}$ as a work storage system initially in the state $\ket{0}$, by definition there exists a different set of thermal operations that extracts work $W_{\sigma \rightarrow \tau}$ from $\sigma$
\begin{equation} \label{eq:sigmawork}
\sigma \otimes \ket{0}\bra{0} \stackrel{TO}{\longrightarrow} \tau \otimes \ket{W_{\sigma \rightarrow \tau}}\bra{W_{\sigma \rightarrow \tau}}.
\end{equation}
By linearity, applying this set of TO to $\tau=p^* \sigma + (1-p^*) X$ yields:
\begin{equation}
p^* \tau \otimes \ketbra{W_{\sigma \rightarrow \tau}}{W_{\sigma \rightarrow \tau}} + (1-p^*) X'_{sw},
\end{equation}
where $X'_{SW}$ is some joint system-weight state, with the weight in some work distribution $p_X(w)$. Note that this operation is applied on both system and weight, and does not need to conserve the thermal state of the system alone.
The Jarzynski equality for this operation reads:
\begin{equation}
p^* e^{\beta W_{\sigma \rightarrow \tau}} + (1-p^*)\sum_{w}  p_X(w) e^{\beta w}=1.
\end{equation}
The second term in this sum is positive, and hence we have:
\begin{equation}
p^* e^{\beta W_{\sigma \rightarrow \tau}} \le 1,
\end{equation}
which is the upper bound of Lemma \ref{thermobounds}.

Note that in situations where the upper bound is saturated (such as reversible processes with $W_{\sigma \rightarrow \tau}=-W_{\tau \rightarrow \sigma}$, when the thermomajorization curve of $\sigma$ is also a straight line) the operation in Eq. \eqref{eq:sigmawork} costs a divergent amount of work in the case of failure i.e. from the state $X$ in Eq. \eqref{eq:tautosigma}.

\subsection{Changing Hamiltonian} \label{sec:Ham Change}

Our results so far have assumed that $\rho$ and $\sigma$ are associated with the same Hamiltonian. Suppose the initial system has Hamiltonian $H_1$ and the final system Hamiltonian $H_2$. Following \cite{HO-limitations}, this scenario can be mapped to one with identical initial and final Hamiltonian, $H$, if we instead consider the transition between $\rho\otimes\ketbra{0}{0}$ and $\sigma\otimes\ketbra{1}{1}$ where:
\begin{equation}
H=H_1\otimes\ketbra{0}{0}+H_2\otimes\ketbra{1}{1}.
\end{equation}
Note that the partition function associated with $H$ is $Z=Z_1+Z_2$.

The height of the thermo-majorization curve of $\rho\otimes\ketbra{0}{0}$ with respect to $H$, is identical to that of $\rho$ with respect to $H_1$ on $[0,Z_1]$ and equal to 1 on $[Z_1,Z]$. Similarly, the height of the thermo-majorization curve of $\sigma\otimes\ketbra{1}{1}$ is identical to that of $\sigma$ on $[0,Z_2]$ and equal to 1 on $[Z_2,Z]$. Hence by extending the definition of $\tilde{V}_{x}\left(\rho\right)$ so that $\tilde{V}_{x}\left(\rho\right)=1$ for $x\geq Z_1$, we can readily apply Theorems \ref{Finite set theorem} and \ref{pthermal} to the case of changing Hamiltonians.

Note that as $\tilde{L}_y\left(\rho\right)=\tilde{L}_y\left(\rho\otimes\ketbra{0}{0}\right)$ for $0\leq y\leq1$ (and similarly for $\sigma$), changing Hamiltonians does not affect the results of Section \ref{sec: TO work of trans}.

\section{Conclusion}
\label{sec:conclusion}

Here, we have introduced a finite set of functions which, like the free energy, can only go down in the resource theory of Thermal Operations. We used these to compute the work of transition, and the maximum probability of making a transition between two states. Finally, we saw that the work of transition between the two states, and vice-versa, can be used to bound the maximum probability of making the transition.

In maximizing the value of $p$ in Eq. \eqref{Goal} to obtain $p^*$, we have attempted to maximize the fraction of $\sigma$ present in a state obtainable from $\rho$. With access to a single two outcome measurement, $\sigma$ can also be obtained from $\rho$ with probability at least $p^*$. There are other measures that one could quantify in attempting to obtain a state that behaves like $\sigma$. For example, one could consider the fidelity between $\sigma$ and a state reachable from $\rho$:
\begin{equation}
F_{\textit{TO}}\left(\rho,\sigma\right)\equiv\max_{\tilde{\rho}} \left\{ F\left(\tilde{\rho},\sigma\right): \rho\stackrel{\textit{TO}}{\longrightarrow}\tilde{\rho} \right\},
\end{equation}
where $F\left(\tilde{\rho},\sigma\right)=\tr\left[\sqrt{\sqrt{\tilde{\rho}}\sigma\sqrt{\tilde{\rho}}}\right]$ is the fidelity between the two states. Investigating this problem is an open question, but note that for diagonal $\sigma$ we have $F_{\textit{TO}}\left(\rho,\sigma\right)\geq F\left(\rho',\sigma\right)\geq\sqrt{p^*}$.

Another alternative would be to consider \emph{heralded} probabilistic transformations. Here a 2-level flag system with trivial Hamiltonian and starting in the state $\ket{0}$ is provided with the initial state $\rho$. The goal is to transform both system and flag so that a measurement on the final flag state would reveal that the system is in state $\sigma$ with probability $p$ and some other state with probability $1-p$. More concretely, one would be interested in maximizing the value of $p$ in the transformation:
\begin{equation} \label{eq:herald goal}
\rho\otimes\ketbra{0}{0}\stackrel{TO}{\longrightarrow}\hat{\rho}=p\sigma\otimes\ketbra{0}{0}+\left(1-p\right)X\otimes\ketbra{1}{1}.
\end{equation}
Due to the results in Section \ref{sec:measurement} it is clear that the maximum value of $p$ achievable in the heralded case Eq.~\eqref{eq:herald goal}, is at least at large as $p^*$ in the unheralded case Eq.~\eqref{Goal} for block-diagonal $\sigma$. 
In follow-up work to the initial version of this manuscript the converse was proven \cite{renes2015relative} and thus the two maximum probabilities are equal. In Appendix~\ref{ap:herald} we extend this analysis to consider the achievable heralded probability when $\sigma$ contains coherences or when one may use a catalyst to assist in the transformation.

At the moment, although our results regarding maximum extractable work are general, little is known about transitions when the final state is not block-diagonal in the energy eigenbasis. In such a situation, our results provide necessary conditions but are not sufficient. Finding sufficient conditions is expected to be difficult, as we do not know such conditions even for non-probabilistic transformations. For recent results on the role of coherences in quantum thermodynamics, see for example \cite{cwiklinski2015limitations,lostaglio2015description, lostaglio2014quantum, narasimhachar2015low}. Nonetheless, we are able to utilize some of these results to provide bounds on the achievable heralded probability when the target state contains coherences in energy. This is done in Appendix~\ref{ap:herald} .

Our analysis has focused on Noisy and Thermal Operations in the absence of a catalyst, i.e. an ancilla which is used to aid in a transition but returned in the same state. In \emph{Catalytic Thermal Operations}, CTO, given $\rho$ and $\sigma$, we are interested in whether there exists a state $\omega$ such that:
\begin{equation}
\rho\otimes\omega\stackrel{\textit{TO}}{\longrightarrow}\sigma\otimes\omega.
\end{equation} 
If such an $\omega$ exists, we say $\rho\stackrel{CTO}{\longrightarrow}\sigma$. There exist instances where $\rho{\stackrel{\textit{TO}}{\longarrownot\longrightarrow}}\sigma$ and yet $\rho\stackrel{\textit{CTO}}{\longrightarrow}\sigma$. Investigating when such catalytic transitions exists has led to a family of second laws of thermodynamics that apply in the single-shot regime \cite{brandao2013second}. Having access to catalysts has the potential to achieve higher values of $p$ than that defined by $p^*$ and it would be interesting to find an expression or bound for the maximum value of $p$ in the process:
\begin{equation}
\rho\stackrel{CTO}{\longrightarrow}\rho'=p\sigma+\left(1-p\right)X.
\end{equation}  
Note that a bound can be obtained from any non-increasing monotone of CTO, $M$ say, that satisfies $M\left(p\sigma+\left(1-p\right)X\right)\geq p M\left(\sigma\right)$. Bounding the maximum transition probability under Catalytic Thermal Operations is made more difficult by the fact that the generalized free energies found in \cite{brandao2013second} are not concave. However, for the case of heralded probability, the situation is somewhat easier and in Appendix~\ref{ap:herald}, we completely characterize what is achievable under CTO when the target state is block-diagonal in energy.

Another avenue of research is to generalize our result to the case where one is interested in not only maximizing the probability of obtaining a single state, but rather, finding the probability simplex of going to an ensemble of many states. Again, the fact that the monotones used in thermodynamics are not in general concave, means that straight application of the techniques used in entanglement theory \cite{jonathan1999minimal} cannot be immediately applied.

Finally, by supplying more work or demanding that extra work is extracted, the value of $p^*$ achieved can be raised or lowered. For $W\leq0$, one could calculate $p^*$ (as a function of $W$) for the states $\rho\otimes\ketbra{W}{W}$ and $\sigma\otimes\ketbra{0}{0}$. For $W>0$ the states to consider would be $\rho\otimes\ketbra{0}{0}$ and $\sigma\otimes\ketbra{W}{W}$. What is the tradeoff between $p^*$ and $W$? As an example, the solution for qubit systems in the Noisy Operations framework is given in Appendix \ref{ap:tradeoff}.

This work has focused on the probability with which a given state can fluctuate into another under a thermodynamical process. The term fluctuation is usually applied within thermodynamics to the concept of fluctuating work, a notion most famously captured by the Jarzynski equality and Crooks' theorem. These were derived under the framework of stochastic thermodynamics while our research was based on applying ideas from quantum information theory. Finding common ground between the two paradigms is likely to be beneficial to both fields and links between work-based fluctuation theorems and the resource theory operation have been developed in \cite{halpern2015introducing, salek2015fluctuations}. In work related to this paper \cite{alhambra2016second}, we shall strengthen these connections still further, formulating the idea of fluctuating work within the resource theory approach and providing new insight into the associated fluctuation theorems. What is more, we shall find fully quantum generalizations and see how the 2nd law of thermodynamics can be recast as an equality.


\begin{acknowledgments} 
We thank Fernando Brand\~{a}o, Micha\l~Horodecki and Lluis Masanes for interesting discussions, and especially Antonio Ac\'{i}n whose initial question sparked the present research. We also thank Felix Leditzky for comments on an earlier version. JO is supported by a Royal Society Wolfson Merit Award and by an EPSRC established career Fellowship.
\end{acknowledgments}

\bibliographystyle{apsrev4-1}
\bibliography{References,common/refthermo,common/work,common/refjono2,common/rmp12,common/refgrav2}

\clearpage
\widetext
\appendix
\section{Entanglement cost of transformations under LOCC} \label{ap:LOCC}

The monotones that we have used for studying Noisy Operations, have been, or can be, defined solely in terms of Lorenz curves. They are also monotones in the resource theory of bipartite pure state entanglement manipulation under Local Operations and Classical Communication \cite{nielsen1999conditions, vidal2000entanglement}, where such curves can also be constructed. Using our monotones, and the behavior of Lorenz curves under tensor product with certain states, we give an expression for the single-shot \emph{entanglement of transition}. This is the amount of entanglement that must be added (or can be extracted) in transforming $\ket{\Psi_{AB}}$ into $\ket{\Phi_{AB}}$ under LOCC.

Previous work has considered the distillable entanglement and entanglement cost - the entanglement of transition when one of $\ket{\Phi_{AB}}$ or $\ket{\Psi_{AB}}$, respectively, is taken to be a separable state. In \cite{buscemi2010distilling}, the amount of entanglement that can be distilled from a single copy of a bipartite mixed state, $\sigma_{AB}$, was bounded in terms of the coherent information. For a bipartite pure state, $\ket{\Psi_{AB}}$, it is given precisely by the min-entropy of the reduced state $\tr_B\ketbra{\Psi_{AB}}{\Psi_{AB}}$ \cite{buscemi2010general}. The amount of entanglement required to create a single copy of $\sigma_{AB}$ was calculated in \cite{buscemi2011entanglement} in terms of the conditional zero-R\'{e}nyi entropy. In each paper, the analysis extends to accomplishing the task up to fixed error, $\epsilon$. Here we go beyond the distillation and cost, showing that the more general entanglement of transition between two arbitrary pure bipartite states, can be quantified in terms of the monotones $L_y$.

For a bipartite pure state, $\ket{\Psi}$, on a system $AB$, let:
\begin{equation}
\rho_{\ket{\Psi}}=\tr_B\ketbra{\Psi}{\Psi}.
\end{equation} 
Without access to any additional resources, it is possible for two separated parties to transform $\ket{\Psi}$ into another bipartite state, $\ket{\Phi}$, under LOCC if and only if $\rho_{\ket{\Phi}}$ majorizes $\rho_{\ket{\Psi}}$ \cite{nielsen1999conditions}. Hence if $\ket{\Psi}$ can be transformed into $\ket{\Phi}$:
\begin{equation}
V_l \left(\rho_{\ket{\Phi}}\right) \geq V_l\left(\rho_{\ket{\Psi}}\right), \quad \forall l,
\end{equation}
and:
\begin{equation}
L_y\left(\rho_{\ket{\Phi}}\right) \leq L_y\left(\rho_{\ket{\Psi}}\right), \quad \forall y\in\mathcal{D}\left(\rho_{\ket{\Psi}}\right),
\end{equation}
where the functions $V_l$, $L_y$ and the set $\mathcal{D}$ are defined as per Section \ref{sec:NO}. Note that for LOCC we consider the `elbows' of the Lorenz curve associated with the initial state whilst for NO we consider the `elbows' of the final state's curve when determining if a transition is possible. This change occurs as for a transition to take place in pure state entanglement theory, we require that the final state majorizes the initial state whilst in the theory of NO, we require that the initial state majorizes the final.

The unit for quantifying entanglement costs is the ebit - the maximally entangled state with local dimension 2. The maximally entangled state with local dimension $d$:
\begin{equation}
\ket{e_d}=\frac{1}{\sqrt{d}}\sum_{i=0}^{d-1}\ket{i}_A\ket{i}_B,
\end{equation}
requires the two parties to share $\log d$ ebits to prepare it and they can extract $\log d$ shared ebits if they share one. Separable states are free within this resource theory so if we define:
\begin{equation}
\ket{\textrm{sep}_d}=\ket{0}_A\ket{0}_B,
\end{equation}
as a separable pure state with local dimension $d$, $\ket{\textrm{sep}_d}$ costs 0 ebits to prepare and no shared entanglement can be extracted from it. Note that:
\begin{align}
L_y\left(\rho_{\ket{\Psi}\otimes\ket{e_d}}\right)&=L_y\left(\rho_{\ket{\Psi}}\right),\\
L_y\left(\rho_{\ket{\Psi}\otimes\ket{\textrm{sep}_d}}\right)&=\frac{1}{d}L_y\left(\rho_{\ket{\Psi}}\right).
\end{align}

The entanglement of transition, $E_{\ket{\Psi}\rightarrow\ket{\Phi}}$, is the optimal amount of shared, bipartite entanglement that the parties need to add, or can gain, to transform a copy of $\ket{\Psi}$ into $\ket{\Phi}$ under LOCC. If the quantity is negative, entanglement must be used up to make the transition possible while if it is positive, entanglement can be extracted. $E_{\ket{\Psi}\rightarrow\ket{\Phi}}$ is the maximum value of $v\log d_2 - u\log d_1$ that can be achieved where $u,v,d_1,d_2\in\mathbb{Z}$ are such that:
\begin{equation}
\ket{\Psi}\ket{e_{d_1}}^{\otimes u} \ket{\textrm{sep}_{d_2}}^{\otimes v} \stackrel{\textit{LOCC}}{\rightarrow} \ket{\Phi}\ket{e_{d_2}}^{\otimes v} \ket{\textrm{sep}_{d_1}}^{\otimes u}.
\end{equation}
In terms of Lorenz curves, the addition of entangled and separable state serve to rescale (with respect to the $x$-axis) the curves associated with $\ket{\Psi}$ and $\ket{\Phi}$ by ${d_2}^{-v}$ and ${d_1}^{-u}$ respectively. To maximize $E_{\ket{\Psi}\rightarrow\ket{\Phi}}$, the Lorenz curve of the rescaled $\ket{\Psi}$ needs to lie just to the right of the Lorenz curve of the rescaled $\ket{\Phi}$. Hence:
\begin{align}
\frac{1}{{d_2}^v} L_{y}\left(\rho_{\ket{\Psi}}\right) \geq \frac{1}{{d_1}^u} L_{y}\left(\rho_{\ket{\Phi}}\right), \quad \forall y\in\mathcal{D}\left(\rho_{\ket{\Psi}}\right),
\end{align}
with equality for some $y$. This gives:
\begin{align}
2^{-\left(E_{\ket{\Psi}\rightarrow\ket{\Phi}}\right)}=\frac{{d_1}^u}{{d_2}^v}=\max_{y\in\mathcal{D}\left(\rho_{\ket{\Psi}}\right)}\frac{L_{y}\left(\rho_{\ket{\Phi}}\right)}{L_{y}\left(\rho_{\ket{\Psi}}\right)},
\end{align}
in analogy with Lemma \ref{NO Work Exp} for the work of transition in Noisy Operations.

This can be generalized to consider situations where we require only that the final state is $\epsilon$-close to the target state $\Phi$ with respect to a measure such as the squared fidelity, $F^2\left(\ket{\Phi'},\ket{\Phi}\right)=\left|\braket{\Phi'}{\Phi}\right|^2$. Let:
\begin{equation}
b^{\epsilon}\left(\ket{\Phi}\right)=\left\{\ket{\Phi'}:\left|\braket{\Phi'}{\Phi}\right|^2\geq1-\epsilon\right\}.
\end{equation}
Then, defining $E^{\epsilon}_{\ket{\Psi}\rightarrow\ket{\Phi}}$ by:
\begin{equation}
E^{\epsilon}_{\ket{\Psi}\rightarrow\ket{\Phi}}=\max_{\ket{\Phi'}\in b^{\epsilon}\left(\ket{\Phi}\right)} E_{\ket{\Psi}\rightarrow\ket{\Phi'}},
\end{equation}
we can write:
\begin{equation}
E^{\epsilon}_{\ket{\Psi}\rightarrow\ket{\Phi}}=\max_{\ket{\Phi'}\in b^{\epsilon}\left(\ket{\Phi}\right)}\left\{-\log\left[\max_{y\in\mathcal{D}\left(\rho_{\ket{\Psi}}\right)}\frac{L_{y}\left(\rho_{\ket{\Phi'}}\right)}{L_{y}\left(\rho_{\ket{\Psi}}\right)}\right]\right\}.
\end{equation}

\newpage

\section{Heralded probability} \label{ap:herald}

In this work we have considered the optimization of $p$ in the process:
\begin{equation} \label{eq:p mix}
\rho\stackrel{TO}{\longrightarrow}\rho'= p\sigma+\left(1-p\right)X,
\end{equation}
for given $\rho$ and $\sigma$. Another related notion of a probabilistic transformation is that of \emph{heralded probability}, i.e. a conclusive fluctuation to a state. In this setup, a qubit flag system with trivial Hamiltonian $H_F\propto\mathbb{I}$ is incorporated which starts in the state $\ket{0}$ and after the Thermal Operation, indicates whether the system was successfully transformed into $\sigma$. More concretely, with respect to heralded probability and for given $\rho$ and $\sigma$ one would attempt to maximize $p$ in the process:
\begin{equation} \label{eq:p flag}
\rho\otimes\ketbra{0}{0}\stackrel{TO}{\longrightarrow}\hat{\rho}=p\sigma\otimes\ketbra{0}{0}+\left(1-p\right)X\otimes\ketbra{1}{1},
\end{equation}
where the total Hamiltonian is $H=H_S+H_F$. A measurement on the flag will result in the system being in state $\sigma$ with probability $p$ and state $X$ with probability $1-p$.

When $\sigma$ is block-diagonal in the energy eigenbasis, the measurement strategy given in Section \ref{sec:measurement} can be used to convert a protocol obtaining a value of $p$ in Eq.~\eqref{eq:p mix} into one that obtains a value of $p$ in Eq.~\eqref{eq:p flag}. Indeed, since our initial manuscript, it has been shown that the maximum value of $p$ that can be achieved in both scenarios for such $\sigma$ is identical\cite{renes2015relative}.

However, analyzing the optimization of $p$ in Eq.~\eqref{eq:p flag} is more tractable than the equivalent problem with respect to Eq.~\eqref{eq:p mix} as for the problem of heralded probability we may always take $X=\tau_S$, the thermal state of the system. To see this, assume that we start with the state-Hamiltonian pair:
\begin{equation}
\left(p\sigma\otimes\ketbra{0}{0}+\left(1-p\right)X\otimes\ketbra{1}{1},H_S+H_F\right),
\end{equation}
and then apply the following Thermal Operations:
\begin{enumerate}
\item Append a thermal state with Hamiltonian $H_B=H_S$:
\begin{align*}
&\left(p\sigma\otimes\ketbra{0}{0}+\left(1-p\right)X\otimes\ketbra{1}{1},H_S+H_F\right)\\
\stackrel{TO}{\longrightarrow}&\left(p\sigma\otimes\tau_B\otimes\ketbra{0}{0}+\left(1-p\right)X\otimes\tau_B\otimes\ketbra{1}{1},H_S+H_B+H_F\right).
\end{align*}
\item Apply the unitary $U=\mathbb{I}_{SB}\otimes\ketbra{0}{0}+U^{\textit{swap}}_{SB}\otimes\ketbra{1}{1}$ where $U^{\textit{swap}}_{SB}$ is the unitary that swaps the state of the system with the state of the bath. As $H_S=H_B$, $\left[U,H_S+H_B+H_F\right]=0$ and hence $U$ is a valid Thermal Operation. This implements:
\begin{align*}
&\left(p\sigma\otimes\tau_B\otimes\ketbra{0}{0}+\left(1-p\right)X\otimes\tau_B\otimes\ketbra{1}{1},H_S+H_B+H_F\right)\\
\stackrel{TO}{\longrightarrow}&\left(p\sigma\otimes\tau_B\otimes\ketbra{0}{0}+\left(1-p\right)\tau_S\otimes X\otimes\ketbra{1}{1},H_S+H_B+H_F\right).
\end{align*}
\item Discard the bath system:
\begin{align*}
&\left(p\sigma\otimes\tau_B\otimes\ketbra{0}{0}+\left(1-p\right)\tau_S\otimes X\otimes\ketbra{1}{1},H_S+H_B+H_F\right)\\
\stackrel{TO}{\longrightarrow}&\left(p\sigma\otimes\ketbra{0}{0}+\left(1-p\right)\tau_S\otimes\ketbra{1}{1},H_S+H_F\right).
\end{align*}
\end{enumerate}
Hence, given a state of the form $\hat{\rho}$, we can always find a Thermal Operation that converts $X$ into $\tau_S$. In attempting to maximize $p$ in Eq.~\eqref{eq:p flag} we can thus always assume that $X$ is the thermal state of the system. This simplification will enable us to prove additional bounds on the maximum value of the heralded probability, $\hat{p}$, for Catalytic Thermal Operations and the case where $\sigma$ contains coherences in energy. 

\subsection{Heralded probability with catalysts}

In Catalytic Thermal Operations, given $\rho$ and $\sigma$, we are interested in whether there exists a state $\omega$ such that:
\begin{equation}
\rho\otimes\omega\stackrel{TO}{\longrightarrow}\sigma\otimes\omega.
\end{equation}
If such an $\omega$ exists, we say it catalyzes the transformation and write $\rho\stackrel{CTO}{\longrightarrow}\sigma$. Determining whether such an $\omega$ exists has resulted in a family of second laws of thermodynamics \cite{brandao2013second}.

Defining the generalized free energies of $\left(\rho,H_S\right)$ by:
\begin{equation}
F_{\alpha}\left(\rho||\tau_S\right)=kT D_\alpha\left(\rho||\tau_S\right)-kT\log Z_S,
\end{equation}
where $D_\alpha$ are the R\'{e}nyi divergences given by:
\begin{equation}
D_\alpha\left(\rho||\tau_S\right)=\frac{\textrm{sgn}\left(\alpha\right)}{\alpha-1}\log\tr\left[\rho^\alpha\tau_S^{1-\alpha}\right],
\end{equation}
then for block-diagonal $\sigma$, $\rho\stackrel{CTO}{\longrightarrow}\sigma$ if and only if $F_{\alpha}\left(\rho_D||\tau_S\right)\geq F_{\alpha}\left(\sigma||\tau_S\right)$, holds $\forall\alpha\geq0$. If $\sigma$ is not block-diagonal, then by replacing $\sigma$ with $\sigma_D$ in these expressions we obtain conditions that are necessary but not sufficient.

To optimize the heralded probability of a transformation from $\rho$ to $\sigma$ under Catalytic Thermal Operations, we thus want to maximize the value of $p$ in $\hat{\rho}=p\sigma\otimes\ketbra{0}{0}+\left(1-p\right)\tau_S\otimes\ketbra{1}{1}$ subject to these free energy constraints applied to $\rho$ and $\hat{\rho}$. This gives us:
\begin{equation}
\hat{p}\leq\max\left\{p: F_\alpha\left(\rho_D\otimes\ketbra{0}{0}\Big|\Big|\tau_S\otimes\mathbb{I}_2\right)\geq F_\alpha\left(p\sigma_D\otimes\ketbra{0}{0}+\left(1-p\right)\tau_S\otimes\ketbra{1}{1}\Big|\Big|\tau_S\otimes\mathbb{I}_2\right),\quad \alpha\in\left[0,\infty\right] \right\}.
\end{equation}
Furthermore, when $\sigma$ is block diagonal in the energy eigenbasis, this bound on $\hat{p}$ is achievable as the second laws \cite{brandao2013second} imply there exists an $\omega$ such that:
\begin{equation}
\rho\otimes\ketbra{0}{0}\otimes\omega\stackrel{TO}{\longrightarrow}\left(\hat{p}\sigma\otimes\ketbra{0}{0}+\left(1-\hat{p}\right)\tau_S\otimes\ketbra{1}{1}\right)\otimes\omega.
\end{equation}
 
\subsection{Heralded probability for arbitrary quantum states}

Quantum generalizations of the R\'{e}nyi divergences have also been used to construct constraints on coherence manipulation under Thermal Operations. Specifically, if we define the \emph{free coherence} of a state $\rho$ by:
\begin{equation}
A_\alpha\left(\rho\right)=S_\alpha\left(\rho||\rho_D\right),
\end{equation}
where $S_\alpha$ are the quantum R\'{e}nyi divergences given by:
\begin{equation}
S_\alpha\left(\rho||\rho_D\right)=\left\{\begin{array}{ll}
\frac{1}{\alpha-1}\log\tr\left[\rho^\alpha\rho_D^{1-\alpha}\right],& \alpha\in\left[0,1\right),\\
\tr\left[\rho\left(\log\rho-\log\rho_D\right)\right] , & \alpha=1,\\
\frac{1}{\alpha-1}\log\tr\left[\left(\rho_D^{\frac{1-\alpha}{2\alpha}}\rho\rho_D^{\frac{1-\alpha}{2\alpha}}\right)^\alpha\right],& \alpha>1,
\end{array}\right.
\end{equation}

then it was shown in \cite{lostaglio2015description} that for general $\sigma$, $\rho\stackrel{TO}{\longrightarrow}\sigma$ only if $A_\alpha\left(\rho\right)\geq A_\alpha\left(\sigma\right)$ for all $\alpha\geq 0$.

Using this we obtain the following bound on the maximum heralded probability of a transformation from $\rho$ to $\sigma$ under Thermal Operations:  
\begin{equation}
\hat{p}\leq\max\left\{p: A_\alpha\left(\rho\otimes\ketbra{0}{0}\right)\geq A_\alpha\left(p\sigma\otimes\ketbra{0}{0}+\left(1-p\right)\tau_S\otimes\ketbra{1}{1}\right),\quad \alpha\in\left[0,\infty\right]\right\}.
\end{equation}

\newpage

\section{The tradeoff between probability and work of transition for a qubit under Noisy Operations} \label{ap:tradeoff}

In this appendix we consider how $p^*$ varies if we supply additional work when attempting to convert $\rho$ into $\sigma$. Alternatively we could attempt to extract extra work during the process. Whilst characterizing the behavior of $p^*$ in general is an open question, here we give the solution for qubit systems with trivial Hamiltonian.

Consider two qubits: $\rho$ with ordered eigenvalues $\vec{\eta}=\{\eta_1,\eta_2\}$ and $\sigma$ with ordered eigenvalues $\vec{\zeta}=\{\zeta_1,\zeta_2\}$. For the transition:
\begin{align}
\begin{array}{rclc}
\rho\otimes s_{|W|}&\stackrel{NO}{\longrightarrow}&\rho'=p\sigma+\left(1-p\right)X,&\quad\textrm{if }W \le 0,\\
\rho&\stackrel{NO}{\longrightarrow}&\rho'=p\sigma\otimes s_{|W|}+\left(1-p\right)X,&\quad\textrm{if }W>0,
\end{array}
\end{align}
how does $p^*$ behave as a function of $W$? Note that using Theorem \ref{NO Theorem}, $p^*\left(0\right)$ is given by $\min\left\{\frac{\eta_1}{\zeta_1},1\right\}$. For $W\leq W_{\rho\rightarrow\sigma}$, by definition we have that $p^*\left(W\right)=1$ (as for these values of $W$, the transition can be performed deterministically). So as to investigate the behavior of the function at $W=0$, in what follows we shall assume $\eta_1<\zeta_1$ and hence $W_{\rho\rightarrow\sigma}<0$.

First take $W\leq 0$ and for simplicity, assume it can be written as $W=-\log\frac{d}{j}$. Then:
\begin{align}
\rho\otimes s_{|W|}&=\textrm{diag}\biggl(\underbrace{\frac{\eta_1}{j},\hdots,\frac{\eta_1}{j}}_{j},\underbrace{\frac{\eta_2}{j},\hdots,\frac{\eta_2}{j}}_{j},\underbrace{0,\hdots,0\vphantom{\frac{\eta_1}{j}}}_{2\left(d-j\right)}\biggr),\\
\sigma\otimes\frac{\mathbb{I}}{d} &=\textrm{diag}\biggl(\underbrace{\frac{\zeta_1}{d},\hdots,\frac{\zeta_1}{d}}_{d},\underbrace{\frac{\zeta_2}{d},\hdots,\frac{\zeta_2}{d}}_{d}\biggr).
\end{align}
We now use Theorem \ref{NO Theorem} together with the fact that $p^*\left(W\right)$ will occur at an `elbow' of $\sigma$ (which is equivalent to $\sigma\otimes\frac{\mathbb{I}}{d}$ under Noisy Operations). As $W_{\rho\rightarrow\sigma}<W$, and the transition does not happen with certainty, we need to only consider the elbow $l=d$ in Theorem \ref{NO Theorem}. Thus:
\begin{align}
p^*\left(W\right)=\frac{V_d\left(\rho\otimes s_{|W|}\right)}{V_d\left(\sigma\otimes\frac{\mathbb{I}}{d}\right)} =\frac{\eta_1+\frac{d-j}{j}\eta_2}{\zeta_1}, \quad W_{\rho\rightarrow\sigma}<-\log\frac{d}{j}\leq0.
\end{align}
This can be rearranged to give:
\begin{equation}
p^*\left(W\right)=\left(2-2^{-W}\right)p^*\left(0\right)+\frac{2^{-W}-1}{\zeta_1}, \quad\quad W_{\rho\rightarrow\sigma}< W\leq0.
\end{equation}

Now take $W\geq 0$ and assume it can be written as $W=\log\frac{d}{j}$. Then:
\begin{align}
\rho\otimes\frac{\mathbb{I}}{d} &=\textrm{diag}\biggl(\underbrace{\frac{\eta_1}{d},\hdots,\frac{\eta_1}{d}}_{d},\underbrace{\frac{\eta_2}{d},\hdots,\frac{\eta_2}{d}}_{d}\biggr),\\
\sigma\otimes s_{|W|}&=\textrm{diag}\biggl(\underbrace{\frac{\zeta_1}{j},\hdots,\frac{\zeta_1}{j}}_{j},\underbrace{\frac{\zeta_2}{j},\hdots,\frac{\zeta_2}{j}}_{j},\underbrace{0,\hdots,0\vphantom{\frac{\eta_1}{j}}}_{2\left(d-j\right)}\biggr).
\end{align}
There are two `elbows' on $\sigma\otimes s_{|W|}$, at $l=j$ and $l=2j$. Calculating the ratio of the monotones at these points gives:
\begin{align}
\frac{V_j\left(\rho\otimes\frac{\mathbb{I}}{d}\right)}{V_j\left(\sigma\otimes s_{|W|}\right)}&=\frac{j\frac{\eta_1}{d}}{\zeta_1}=\frac{\eta_1}{\zeta_1} 2^{-W},\label{eq:l=j}\\
\frac{V_{2j}\left(\rho\otimes\frac{\mathbb{I}}{d}\right)}{V_{2j}\left(\sigma\otimes s_{|W|}\right)}&=\left\{
\begin{array}{rlr}
2j\frac{\eta_1}{d}&=2\eta_1 2^{-W} &\text{if }  2j\leq d,\\
\eta_1+\frac{2j-d}{d}\eta_2&=\left(2\eta_1-1\right)+2\left(1-\eta_1\right) 2^{-W} &\text{if } 2j\geq d.
\end{array}\right. \label{eq:l=2j}
\end{align}
It is easy to see that $\frac{\eta_1}{\zeta_1}\leq 2\eta_1$ since $\zeta_1\geq\frac{1}{2}$. Comparing Eq. (\ref{eq:l=j}) with the second case in Eq. (\ref{eq:l=2j}), it is possible to show that:
\begin{equation}
\frac{V_j\left(\rho\otimes\frac{\mathbb{I}}{d}\right)}{V_j\left(\sigma\otimes s_{|W|}\right)} \leq \frac{V_{2j}\left(\rho\otimes\frac{\mathbb{I}}{d}\right)}{V_{2j}\left(\sigma\otimes s_{|W|}\right)} \Leftrightarrow 2^W\geq \frac{\eta_1-2\zeta_1+2\eta_1\zeta_1}{2\eta_1\zeta_1-\zeta_1}.
\end{equation}
As $W\geq0$, the minimum ratio occurs at $l=j$. Hence:
\begin{equation}
p^{*}\left(W\right)=p^{*}\left(0\right)2^{-W},\quad W\geq0.
\end{equation}

Combining these results, we have that for $\eta_1<\zeta_1$:
\begin{equation}
p^{*}\left(W\right)=
\begin{cases}
1 & \text{if } W\leq W_{\rho\rightarrow\sigma},\\
\left(2-2^{-W}\right)p^*\left(0\right)+\frac{2^{-W}-1}{\zeta_1} & \text{if } W_{\rho\rightarrow\sigma}<W\leq 0,\\
p^{*}\left(0\right)2^{-W} &\text{if } 0<W.
\end{cases}
\end{equation}
As an example, in Figure \ref{fig:tradeoff}, we plot $p^*\left(W\right)$ against $W$ for $\vec{\eta}=\{0.6,0.4\}$ and $\vec{\zeta}=\{0.85,0.15\}$.

For completeness, for $\eta_1\geq\zeta_1$:
\begin{equation}
p^{*}\left(W\right)=
\begin{cases}
1 & \text{if } W\leq W_{\rho\rightarrow\sigma},\\
\left(2\eta_1-1\right)+2\left(1-\eta_1\right) 2^{-W} & \text{if } W_{\rho\rightarrow\sigma}<W\leq \log\left(\frac{\eta_1-2\zeta_1+2\eta_1\zeta_1}{2\eta_1\zeta_1-\zeta_1}\right),\\
\frac{\eta_1}{\zeta_1}2^{-W} & \text{if } W> \log\left(\frac{\eta_1-2\zeta_1+2\eta_1\zeta_1}{2\eta_1\zeta_1-\zeta_1}\right).
\end{cases}
\end{equation}

\begin{figure}
\centering 
\includegraphics[width=1\columnwidth]{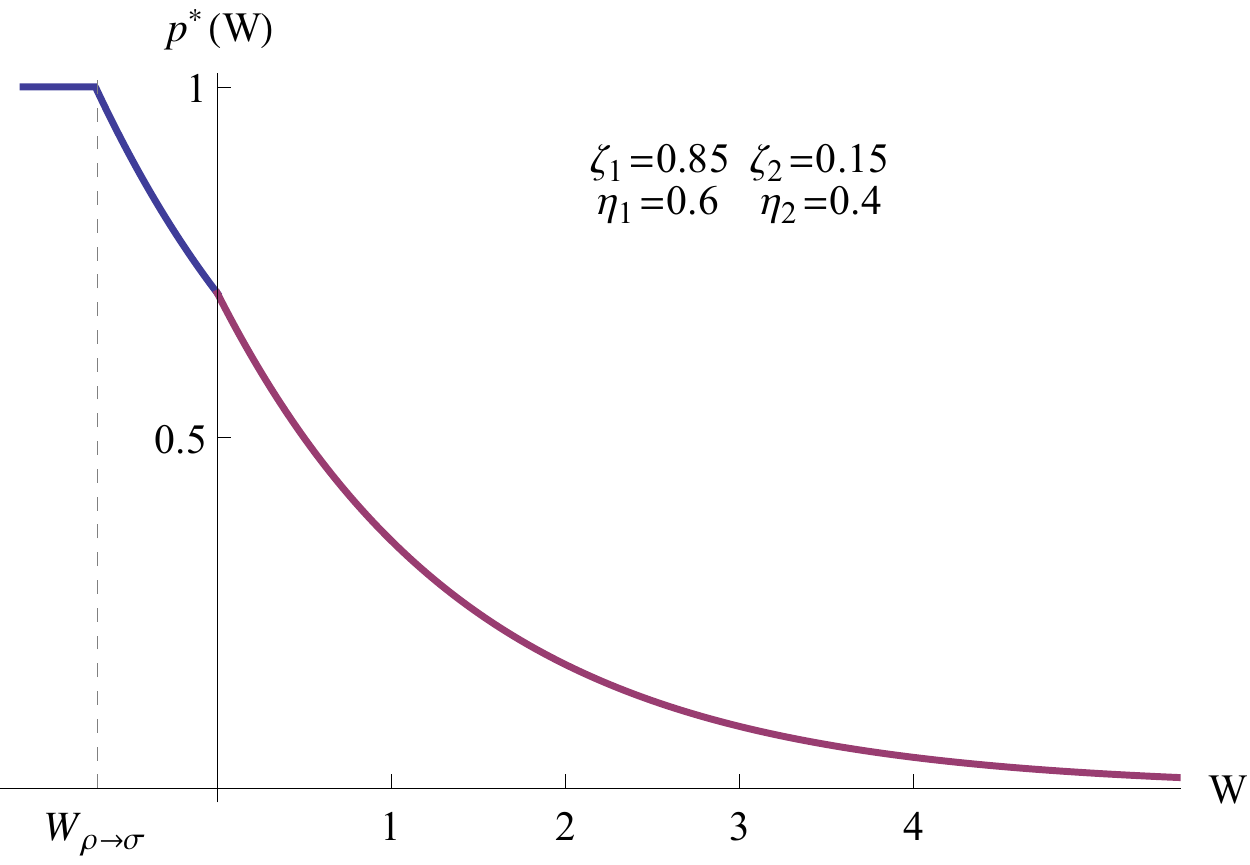} 
\caption{Here we show how $p^*(W)$ varies as a function of $W$ for qubits under Noisy Operations when $W_{\rho\rightarrow\sigma}<0$. Note the behavior at $W=0$, indicating the function is not convex in $W\geq W_{\rho\rightarrow\sigma}$.} \label{fig:tradeoff}
\end{figure}

\newpage

\end{document}